\documentclass[manuscript]{acmart}

\usepackage{booktabs} 
\usepackage{wrapfig}
\usepackage{footmisc}
\usepackage{paralist}
\usepackage{algorithm}
\usepackage{enumitem}
\usepackage[noend]{algpseudocode}
\usepackage{bbm}
\usepackage{hyperref}
\usepackage{amsmath}
\usepackage[most]{tcolorbox}
\usepackage{xspace}
\tcbset{myformula/.style={colback=white, 
		colframe=black, 
		top=0pt,bottom=5pt,left=0pt,right=0pt,
		boxsep=0pt,
		arc=0pt,
		outer arc=0pt,boxrule=0.5pt,
}}
\usepackage{xcolor}
\usepackage{pifont}
\usepackage{multirow}
\usepackage{amsmath}
\usepackage{amsthm}
\usepackage{subcaption}
\usepackage{balance}
\usepackage{makecell}

\newtheorem{problem}{Problem}
\newtheorem{property}{Property}

\definecolor{OliveGreen}{rgb}{0,0.6,0}
\newcommand{\tick}{\textcolor{OliveGreen}{\ding{52}}}

\def\BibTeX{{\rm B\kern-.05em{\sc i\kern-.025em b}\kern-.08em
		T\kern-.1667em\lower.7ex\hbox{E}\kern-.125emX}}

\newcommand{\fgr}[3][\relax]{%
	\begin{figure}[htp]%
		\centering
		\includegraphics[#2]{#3}%
		\vspace{-0.1in}
		\ifx\relax#1\else\caption{{#1}}\fi
	\end{figure}%
}

\newcommand{\method}{{\sc CODEtect}\xspace}
\newcommand{\glocalkd}{{\sc GLocalKD}\xspace}

\newcommand{\gbad}{{\sc Gbad}\xspace}
\newcommand{\iforest}{{\sc iForest}\xspace}
\newcommand{\entropy}{{\sc Entropy}\xspace}
\newcommand{\multi}{{\sc Multiedges}\xspace}

\newcommand{\enron}{{\sc Enron}\xspace}
\newcommand{\sh}{{\sc SH}\xspace}
\newcommand{\shsyn}{{\sc SH\_Synthetic}\xspace}
\newcommand{\kd}{{\sc KD}\xspace}
\newcommand{\hw}{{\sc HW}\xspace}

\newcommand{\cbit}{\begin{compactitem}}
	\newcommand{\ceit}{\end{compactitem}}
\newcommand{\cben}{\begin{compactenum}}
	\newcommand{\ceen}{\end{compactenum}}

\newcommand{\bal}{\begin{align}}
\newcommand{\ean}{\end{align}}
\newcommand{\bit}{\begin{itemize}}
\newcommand{\eit}{\end{itemize}}
\newcommand{\ben}{\begin{enumerate}}
\newcommand{\een}{\end{enumerate}}
\newcommand{\beq}{\begin{equation}}
\newcommand{\eeq}{\end{equation}}

\newcommand{\N}{\mathbb{N}}

\newcommand{\mG}{\mathcal{G}}
\newcommand{\mt}{{MT}}
\newcommand{\smt}{{SMT}}
\newcommand{\mts}{\mathcal{MT}}
\newcommand{\m}{\mathcal{M}}
\newcommand{\mM}{\mathcal{M}}
\newcommand{\mN}{\mathcal{N}}

\newcommand{\mO}{\mathcal{O}}
\newcommand{\mSO}{\mathcal{S}}
\newcommand{\mC}{\mathcal{C}}
\newcommand{\mT}{\mathcal{T}}

\newcommand{\mis}{{MIS}\xspace}
\newcommand{\wmis}{{WMIS}}
\newcommand{\go}{{G_\mathcal{O}}}
\newcommand{\vo}{{V_\mathcal{O}}}
\newcommand{\eo}{{E_\mathcal{O}}}

\newcommand{\coverset}{\mathcal{CS}}

\makeatletter
\renewcommand*\env@matrix[1][*\c@MaxMatrixCols c]{%
	\hskip -\arraycolsep
	\let\@ifnextchar\new@ifnextchar
	\array{#1}}
\makeatother

\newcommand{\ls}{\log^{\star}}

\newcommand{\reminder}[1]{{\textsf{\textcolor{red}{[TODO: #1]}}}}

\newcommand{\hide}[1]{}

\algnewcommand{\LINEIfThenElse}[3]{
	\State \algorithmicif\ #1\ \algorithmicthen\ #2\ \algorithmicelse\ #3}

\algnewcommand\algorithmicforeach{\textbf{for each}}
\algdef{S}[FOR]{ForEach}[1]{\algorithmicforeach\ #1\ \algorithmicdo}

\algdef{SE}[DOWHILE]{Do}{doWhile}{\algorithmicdo}[1]{\algorithmicwhile\ #1}%

\usepackage{array}

\AtBeginDocument{%
  \providecommand\BibTeX{{%
    \normalfont B\kern-0.5em{\scshape i\kern-0.25em b}\kern-0.8em\TeX}}}

\setcopyright{acmcopyright}
\copyrightyear{2022}
\acmYear{2022}

\begin{document}

\title{Detecting Anomalous Graphs in Labeled  Multi-Graph Databases}

\author{Hung T. Nguyen}
\email{hn4@princeton.edu}
\affiliation{%
  \institution{Princeton University}
  \city{Princeton}
  \state{NJ}
  \country{USA}
}

\author{Pierre J. Liang}
\affiliation{%
  \institution{Carnegie Mellon University}
  \city{Pittsburgh}
  \state{PA}
  \country{USA}}
\email{liangj@andrew.cmu.edu}

\author{Leman Akoglu}
\affiliation{%
	\institution{Carnegie Mellon University}
	\city{Pittsburgh}
	\state{PA}
	\country{USA}}
\email{lakoglu@andrew.cmu.edu}

%
%
%
%
%

\renewcommand{\shortauthors}{Hung T. Nguyen, et al.}

\begin{abstract}
  Within a large database $\mG$ containing graphs with labeled nodes and directed, multi-edges; 
how can we 
detect the anomalous graphs? Most existing work are designed for plain (unlabeled) and/or simple (unweighted) graphs.
We introduce \method, the \textit{first} approach that addresses the anomaly detection task for graph databases with such complex nature. To this end, it
identifies a small representative set $\mSO$ of structural patterns (i.e., node-labeled network motifs) that losslessly compress database $\mG$ as concisely as possible. Graphs that do not compress well are flagged as anomalous.
\method exhibits two novel building blocks: ($i$) a motif-based lossless graph encoding scheme, and ($ii$) fast memory-efficient search algorithms for $\mSO$.
We show the effectiveness of \method on transaction graph databases from three different corporations and statistically similar synthetic datasets, where existing baselines adjusted for the task fall behind significantly, across different types of anomalies and performance metrics.


\end{abstract}


\maketitle

\setlength{\textfloatsep}{6pt}

\section{Introduction}
\label{sec:intro}

Given hundreds of thousands of annual transaction records 
of a corporation, how can we identify the abnormal ones, which may indicate entry errors or employee misconduct? 
How can we spot anomalous daily email/call interactions or software programs with bugs?

We introduce a novel anomaly detection technique called \method, for \textit{node-Labeled, Directed, Multi-graph (LDM)} \textit{databases} which emerge from many applications.
Our motivating domain is accounting,
where
each 
transaction record is represented by a graph in which the nodes are accounts and directed edges reflect transactions. Account types (revenue, expense, etc.) are depicted by (discrete) labels and separate transactions between the same pair of accounts create edge multiplicities. 
The problem is then identifying the anomalous graphs within LDM graph databases.
In these abstract terms, \method applies more broadly to other domains exhibiting graph data with such complex nature, e.g., detecting anomalous  
employee email graphs with job titles as labels, call graphs with geo-tags as labels, control flow graphs with function-calls as labels, etc. What is more, \method{} can handle simpler settings with any subset of the LDM properties.

Graph anomaly detection has been studied under various non-LDM settings.
Most of these work focus on detecting anomalies within a \textit{single} graph; either plain (i.e., unlabeled), attributed (nodes exhibiting an {array} of (often continuous) features), or dynamic (as the graph changes over time) \cite{conf/pakdd/AkogluMF10,perozzi2014focused,conf/kdd/EswaranFGM18,shin2016corescope,perozzi2016scalable,hooi2017graph,journals/datamine/AkogluTK15,conf/kdd/ManzoorMA16} (See Table~\ref{tab:related_comparison_t} for overview).
None of these applies to our setting, as we are to detect graph-level anomalies within a graph \textit{database}.
There exist related work for node-labeled graph databases \cite{NobleC03}, which however does not handle multi-edges, and as we show in the experiments (Sec. \ref{sec:experiment}), cannot tackle the problem well.

Recently, general-purpose embedding/representation learning techniques achieve state-of-the-art results in graph classification tasks \cite{grover2016node2vec,niepert2016learning,hamilton2017inductive,narayanan2017graph2vec,fan2018gotcha,peng2018graph}. However, they do not tackle the anomaly detection problem \textit{per se}---the embeddings need to be fed to an off-the-shelf vector outlier detector.
Moreover, most embedding methods \cite{grover2016node2vec,hamilton2017inductive,fan2018gotcha} produce \emph{node} embeddings; how to use those for {graph-level} anomalies is unclear. Trivially aggregating node representations, e.g., by mean or max pooling, to obtain the entire-graph representation provides suboptimal results \cite{narayanan2017graph2vec}. 
Graph embedding techniques \cite{narayanan2017graph2vec,peng2018graph} as well as graph kernels \cite{yanardag2015deep,ShVi09} (paired with a state-of-the-art detector), yield poor performance as we show through experiments (Sec. \ref{sec:experiment}), possibly because embeddings capture general patterns, leaving rare structures out, which are critical for anomaly detection. 

Our main contributions are summarized in the following:

\cbit

\item {\bf Problem Formulation:~} 
Motivated by application to business accounting, 
we consider the anomaly detection problem in labeled directed multi-graph (LDM) databases and propose \method; (to our knowledge) the \textit{first} method to detect anomalous graphs with such complex nature (Sec.~\ref{sec:related}). 
\method also generally applies to simpler, non-LDM settings.
The main idea is to identify a few representative network motifs that are used to encode
the database in a lossless fashion as succinctly as possible. \method then flags those graphs that do not compress well under this encoding as anomalous (Sec.~\ref{sec:problem}). 

\item {\bf New Encoding \& Search Algorithms:~} The graph encoding problem is two-fold: how to encode and which motifs to encode with. To this end, we introduce (1) new lossless motif and graph encoding schemes (Sec.~\ref{sec:encoding}), and (2) efficient search algorithms for identifying key motifs with a goal to minimize the total encoding cost 
(Sec.~\ref{sec:algo}).

\item {\bf Real-world Application:~} In collaboration with industry, 
we apply our proposed techniques to annual transaction records from three different corporations, from small- to large-scale. We show the superior performance of \method over existing baselines in detecting injected anomalies that mimic certain known malicious schemes in accounting.
Case studies on those as well as the public Enron email database further show the effectiveness of \method in spotting noteworthy instances (Sec.~\ref{sec:experiment}). To facilitate reproducibility, we also confirm our performance advantages on statistically similar datasets resembling our real-world databases.

\ceit

\vspace{0.05in}
\noindent
{\bf Reproducibility.~} All source code as well as public-domain and synthetic data are shared at \url{https://bit.ly/2P0bPZQ}. 

\section{Related Work}
\label{sec:related}

\hide{
\method is designed to detect anomalous graphs within a \textit{database} containing graphs with complex properties; such as node labels and/or multi/weighted, and/or directed edges.  To our knowledge, \textit{there exists no other work for this task} that is able to handle graphs with such nature. Table \ref{tab:related_comparison_t} gives a qualitative comparison to existing art, described as follows.	
			
Graph anomaly detection has been the focus of many work \cite{conf/pakdd/AkogluMF10,perozzi2014focused,shin2016corescope,perozzi2016scalable,hooi2017graph} (See \cite{journals/datamine/AkogluTK15} for a survey.) However, these do not apply to detecting anomalies within a graph \textit{database}, as they are designed to find node/edge/subgraph anomalies within a \textit{single} graph.
Several work for detecting anomalies among a set or series of graphs \cite{conf/kdd/EswaranFGM18,conf/kdd/ManzoorMA16,NobleC03} cannot simultaneously handle all the graph properties that \method is designed for, such as node labels or edge weights. 
}

\noindent
{\bf Graph Anomaly Detection:}
Graph anomaly detection has been studied under various settings for plain/attributed, static/dynamic, etc. graphs, including the most recent deep learning based approaches \cite{conf/pakdd/AkogluMF10,perozzi2014focused,conf/kdd/EswaranFGM18,shin2016corescope,perozzi2016scalable,hooi2017graph,ding2019deep,zhao2020error} (See \cite{journals/datamine/AkogluTK15} and {\cite{ma2021comprehensive} for a survey.)
These works focus on detecting \textit{node/edge/subgraph anomalies} within a \textit{single} graph,
none of which applies to our setting, as we are to detect anomalous graphs (or \textit{graph-level anomalies}) within a graph \textit{database}.

On anomalous graph detection in graph databases,
{\sc Gbad} \cite{journals/ida/EberleH07} has been applied to  
flag graphs as anomalous if they experience low compression via discovered substructures 
over the iterations. Further, it has been used to identify graphs
that contain substructures $S'$ with small differences (a few modifications, insertions, or deletions) from the best one $S$, which are attributed to malicious behavior \cite{journals/ida/EberleH07}. {\sc Gbad} also has very high time complexity due to the nested searches for substructures to compress the graphs through many iterations (failed to complete on multiple cases of our experiments - see Section~\ref{sec:experiment}).
Our work is on the same lines with these work in principle, however our encoding scheme is lossless.  Moreover, these work cannot handle graphs with weighted/multi edges.
There exist other graph anomaly detection approaches \cite{conf/kdd/ManzoorMA16,conf/kdd/EswaranFGM18}, however none of them simultaneously handle node-labeled graphs with multi-edges. \textsc{SnapSketch} \cite{paudel2020snapsketch} was introduced recently as an unsupervised graph representation approach for intrusion detection in a graph stream and showed better detection than previous works \cite{conf/kdd/ManzoorMA16,conf/kdd/EswaranFGM18}, however, \textsc{SnapSketch} was originally designed to work on undirected graphs. Note also that these works \cite{conf/kdd/ManzoorMA16,conf/kdd/EswaranFGM18,paudel2020snapsketch} focus on graph streams, i.e., time-ordered graphs, and may not work well in our setting of \textit{unordered} graph databases. We present a qualitative comparison of related work to \method in Table \ref{tab:related_comparison_t}.

\setlength{\tabcolsep}{4pt}

\begin{table}[!t]
	\caption{Comparison with popular approaches to graph anomaly detection, in terms of distinguishing properties.}
		\begin{tabular}{>{\centering\arraybackslash}m{0.08\textwidth} | >{\centering\arraybackslash}m{0.3\textwidth} | >{\centering\arraybackslash}m{.09\linewidth}|>{\centering\arraybackslash}m{.09\linewidth}|>{\centering\arraybackslash}m{.09\linewidth}|>{\centering\arraybackslash}m{.085\linewidth}|>{\centering\arraybackslash}m{0.095\linewidth} }
			\toprule
			\multicolumn{2}{l|}{\textbf{Methods vs. Properties}} & Graph database & Node-labeled & Multi/ Weighted & Directed & Anomaly detection \\
			\midrule
			\multirow{3}{0.1\textwidth}{\textbf{Graph Emb.}} & node2vec\cite{grover2016node2vec}, GraphSAGE\cite{hamilton2017inductive} & & \tick & \tick & \tick & \\ 
			\cline{2-7}
			&graph2vec\cite{narayanan2017graph2vec}, Metagraph2vec\cite{fan2018gotcha}, GIN\cite{DBLP:conf/iclr/XuHLJ19} & \tick & \tick & & \tick & \\
			\cline{2-7}
			&  PATCHY-SAN\cite{niepert2016learning}, MA-GCNN\cite{peng2018graph}, Deep Graph Kernels \cite{yanardag2015deep} & \tick & \tick & \tick & \tick & \\
			\hline
			\multirow{7}{0.1\textwidth}{\textbf{Graph Anom. Detect.}} & OddBall\cite{conf/pakdd/AkogluMF10} & & & \tick & \tick & \tick \\ 
			\cline{2-7}
			& FocusCO\cite{perozzi2014focused}, AMEN\cite{perozzi2016scalable}, \textsc{Dominant}\cite{ding2019deep} & & \tick & & \tick & \tick\\ 
			\cline{2-7}
			& \textsc{GAL}\cite{zhao2020error} & & & \tick & \tick & \tick \\
			\cline{2-7}
			& CoreScope\cite{shin2016corescope} & & & & & \tick \\ 
			\cline{2-7}
			& FRAUDAR\cite{hooi2017graph} & & & &\tick & \tick \\ 
			\cline{2-7}
			& StreamSpot\cite{conf/kdd/ManzoorMA16}, GBAD\cite{journals/ida/EberleH07} & \tick & \tick & & \tick & \tick \\                    
			\cline{2-7}
			& SpotLight\cite{conf/kdd/EswaranFGM18} & \tick & & \tick & \tick & \tick\\ 
			\cline{2-7}
			& \textsc{SnapSketch}\cite{paudel2020snapsketch} & \tick & & \tick & & \tick \\
			\cline{2-7}
			& \textsc{GLocalKD}\cite{ma2022deep} & \tick & \tick & & & \tick \\
			\cmidrule{2-7}\morecmidrules\cmidrule{2-7}
			& \method [this paper] & \tick & \tick & \tick & \tick & \tick \\
			\bottomrule
	\end{tabular}
	\label{tab:related_comparison_t} 
\end{table}

\hide{
Recently, a body of graph embedding methods has been developed, able to handle graphs with complex properties \cite{grover2016node2vec,niepert2016learning,hamilton2017inductive,narayanan2017graph2vec,fan2018gotcha,peng2018graph,DBLP:conf/iclr/XuHLJ19}. Those as well as graph kernels \cite{yanardag2015deep,ShVi09} produce vector representations. However, they do not tackle anomaly detection \textit{per se}. Such representations need to be input to a certain choice of an off-the-shelf detector to perform anomaly detection. While such representations capture general structural patterns, we find they are not suitable for anomaly detection as shown in the experiments.
}

\noindent {\bf Graph Embedding for Anomaly Detection:} Recent graph embedding methods \cite{grover2016node2vec,niepert2016learning,hamilton2017inductive, narayanan2017graph2vec,fan2018gotcha,peng2018graph,ding2019deep,zhao2020error} and graph kernels \cite{yanardag2015deep,ShVi09} find a latent representation of node, subgraph, or the entire graph and have been shown to perform well on classification and link prediction tasks. However, graph embedding approaches, like \cite{grover2016node2vec,hamilton2017inductive,fan2018gotcha}, learn a node representation, which is difficult to use directly for detecting anomalous graphs. Peng et al. \cite{peng2018graph} propose a graph convolutional neural network via motif-based attention, but, this is a supervised method and, thus, not suitable for anomaly detection. Our experimental results show that other recent graph embedding \cite{narayanan2017graph2vec} and graph kernel methods \cite{yanardag2015deep}, that produce a direct graph representation, when combined with a state-of-the-art anomaly detector have low performance and are far less accurate than \method. Concurrent to our work, graph neural networks for anomalous graph detection is studied in \cite{zhao2021using,ma2022deep} that examines end-to-end graph anomaly detection. Further, \cite{bandyopadhyay2020outlier} investigates outlier resistant architectures for graph embedding. A key challenge, in general, for deep learning based models for \textit{unsupervised} anomaly detection is their sensitivity to many hyper-parameter settings (including those for regularization: such as weight decay, and drop-out rate; optimization: such as learning rate, and achitecture: such as depth, width, etc.), which are not straightforward to set in the absence of any ground-truth labels. Distinctly, our work leverages the Minimum Description Length principle and does not exhibit any hyper-parameters.

\noindent
{\bf Graph Motifs:~}
Network motifs have proven useful in understanding the 
functional units and organization of complex systems \cite{milo2002network,Vazquez17940,Benson163}. 
Motifs have also been used as features for network classification \cite{milo2004superfamilies},
community detection \cite{Arenas2008,YiBiKe18}, and
in graph kernels for graph comparison \cite{ShVi09}.
On the algorithmic side, several works have designed fast techniques for identifying significant motifs \cite{KashtanIMA04,Kashani09,BloemR17,bloem2020large}, where a sub-graph is regarded as a motif only if its frequency is higher than expected under a network null model.

Prior works on network motifs mainly focus on 3- or 4-node motifs in undirected unlabeled/plain graphs \cite{widm1226,conf/kdd/ElenbergSBD15,conf/kdd/SaneiMehriST18,conf/www/UganderBK13}, either using subgraph frequencies in the analysis of complex networks or most often developing fast algorithms for counting (e.g., triangles) (See \cite{SeTi19} for a recent survey).
Others have also studied
directed \cite{Benson163}
and
temporal motifs \cite{kovanen2011temporal,paranjape2017motifs}.
Most relatedly, there is a recent work on  
node-labeled subgraphs referred to as heterogeneous network motifs \cite{Rossi19}, 
where again, the focus is on scalable counting.
Our work differs in using heterogeneous motifs as building blocks of a graph encoding scheme, 
toward the goal of anomaly detection.

\hide{
Finally, frequent patterns and lossless compression via the MDL principle \cite{Rissanen78,grunwald2007minimum} have been successfully applied to anomaly detection \cite{conf/sdm/SmetsV11,AkogluTVF12}, however for regular transaction (or vector) databases. We are the first to use motif-based lossless graph encoding for the \textit{graph} anomaly detection task.
}

\noindent
{\bf Data Compression via MDL-Encoding:}
The MDL principle by Rissanen \cite{Rissanen78} states that the best theory to describe a data is the one that minimizes the sum of the size of the theory, and the size of the description of data using the theory. The use of MDL has a long history in itemset mining \cite{conf/icdm/TattiV08,VreekenLS11}, for transaction (tabular) data, 
also applied to anomaly detection \cite{conf/sdm/SmetsV11,AkogluTVF12}.

MDL has also been used for graph compression.
Given a pre-specified list of well-defined structures (star, clique, etc.), it is employed
to find a succinct description of a graph in those ``vocabulary'' terms \cite{KoutraKVF14}. 
This vocabulary is later extended for dynamic graphs \cite{ShahKZGF15}.
A graph is also compressed hierarchically, by sequentially aggregating sets of nodes into super-nodes, where the best summary and associated corrections are found with the help of MDL \cite{NavlakhaRS08}.

There exists some work on  attributed graph compression \cite{TianHP08}, but the goal is to find super-nodes that represent a set of nodes that are homogeneous in some (user-specified) attributes.
{\sc Subdue} \cite{NobleC03} is one of the earliest work
 to employ MDL for substructure discovery in node-labeled graphs.
The aim is to extract the ``best'' substructure $S$ whose encoding plus the encoding of a graph after replacing each instance of $S$ with a (super-)node is as small as possible.
\section{Preliminaries \& The Problem}
\label{sec:problem}

As input, a large set of $J$ graphs $\mG = \{G_1, \ldots, G_J\}$ is given.
Each graph $G_j=(V_j, E_j, \tau)$ is a directed, node-labeled, multi-graph which may contain multiple edges that have the same end nodes.
$\tau: V_j \rightarrow \mT$ is a function that assigns labels from an alphabet $\mT$ to nodes in each graph.
The number of realizations of an edge $(u,v)\in E_j$ is called its \textit{multiplicity}, denoted $m(u,v)$.
(See {Fig. \ref{fig:demo}}(a) for example.)

\begin{figure}[!t]
	\centering
	\includegraphics[width = 0.6\linewidth]{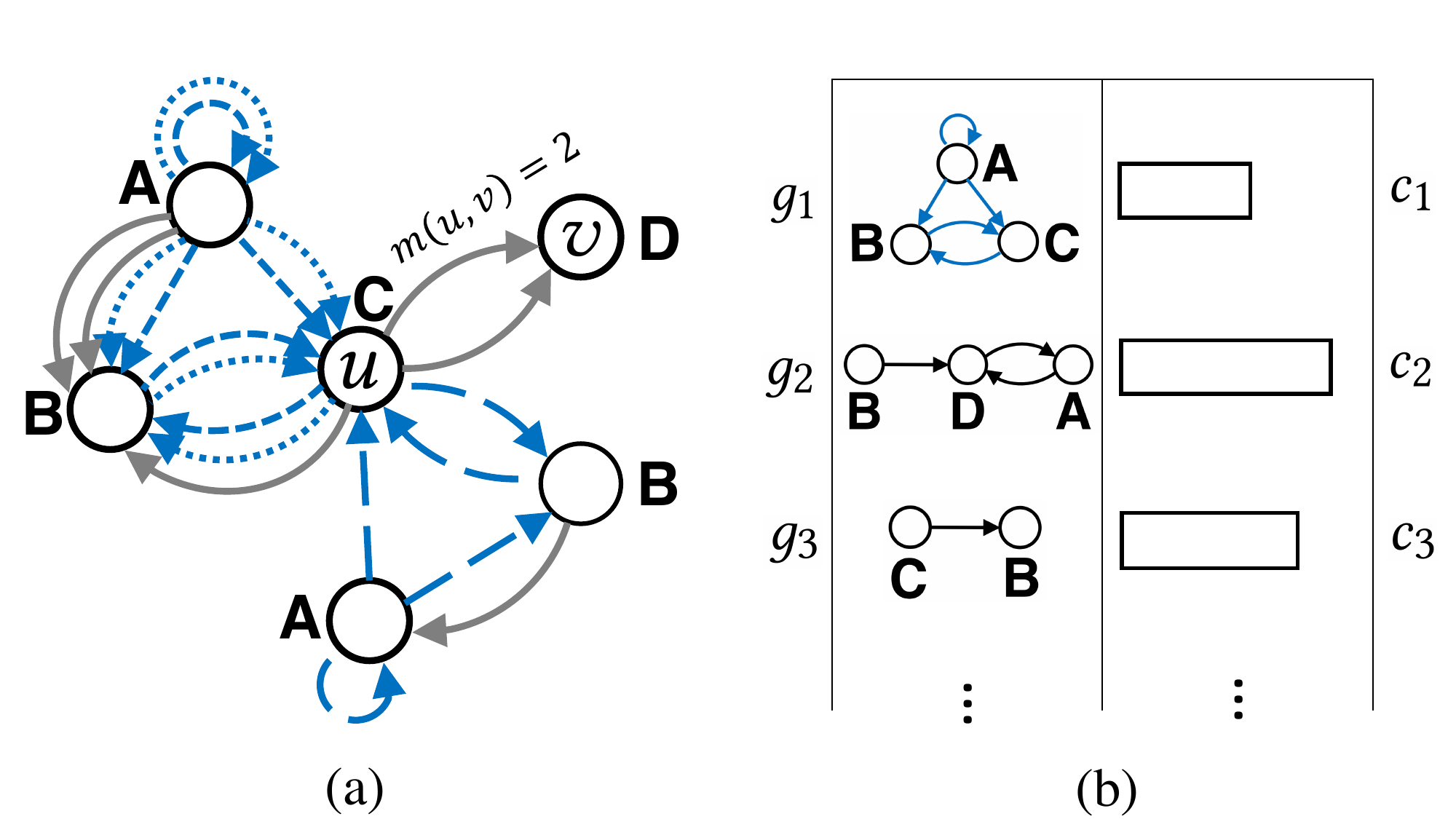}
	\caption{(a) E.g. node-labeled multi-graph where capital letters denote the node labels, and $m(u,v)$ depicts multiplicity of edge $(u,v)$ ; 
	(b) Example motif table; 1st col. lists the motifs, and 2nd col. provides the corresponding codes, 
	bar width depicting code length. The blue edges in (a) sharing the same dash type depict three different occurrences of the top motif (i.e. $g_1$) in the motif table in (b).\label{fig:demo}}
\end{figure}



Our motivating domain is 
business accounting, in which each $G_j$ corresponds to a graph representation of what-is-called a ``journal entry'': a detailed transaction record.
Nodes capture the unique accounts associated with the record, \textit{directed} edges the transactions between these accounts, and node \textit{labels} the financial statement (FS) account types (e.g., assets, liabilities, revenue, etc.). Bookkeeping data is kept as a chronological listing (called General Ledger) of each separate business transaction, 
where multiple transactions involving same account-pairs generate multi-edges between two nodes. 

Our high-level idea for finding anomalous graphs in  database $\mG$ is to identify key characteristic \textit{patterns} of the data that ``explain'' or compress the data well, and flag those graphs that do not exhibit such patterns as expected---simply put, graphs that do not compress well are anomalous.
More specifically, graph patterns are substructures or subgraphs, called \textit{motifs}, which occur frequently within the input graphs.
``Explaining'' the data is encoding each graph using the frequent motifs that it contains. The more frequent motifs we use for encoding, the more we can compress the data; simply by encoding the existence of each such motif with a short code. 

The goal is to find a (small) set of motifs that compresses the data the best.
Building on the Minimum Description Length (MDL) principle \cite{grunwald2007minimum}, we aim to find a model, namely a \textit{motif table} (denoted $\mt$) that contains a carefully selected subset of graph motifs, such that the total code length of (1) the model itself plus (2) the encoding of the data using the model is as small as possible.
In other words, we are after a small model that compresses the data the most.
The two-part objective of minimizing the total code length is given as follows.


\begin{tcolorbox}[ams align,myformula]
	\label{obj}
\underset{{\mt \subseteq \mts}}{\text{minimize}} \;\;\; L(\mt, \mathcal{G}) =  \underbrace{L(\mt)}_{\text{model code length}} \;+\; \underbrace{L(\mathcal{G} | \mt )}_{\text{data code length}} \;\;,
\end{tcolorbox}

\noindent
where $\mts$ denotes the set of all possible candidate motif tables.
The first term can be seen as a model regularizer that penalizes using an unneccesarily large set of motifs to explain the data.
The second term is the compression length of the data with the (selected) motifs and decomposes as $L(\mathcal{G} | \mt) = \sum_j L({G_j} | \mt)$ since individual journals are independent. The encoding length $L({G_j} | \mt)$ is also the {\em anomaly score} for the $j$th graph---the larger, the more anomalous. 

As such, we have a combinatorial subset selection problem toward optimizing Eq.~\eqref{obj}. To this end,
we address two subproblems outlined below. 

\begin{problem}
	Our graph encoding problem is two-fold: (1) how to encode, and (2) which motifs to encode with, or,
	\cben
\item \textbf{Encoding Schemes} (Sec. \ref{sec:encoding}): 
Define schemes for ($i$) $L(\mt)$; encoding the motifs in $\mt$, and ($ii$) $L({G_j} | \mt)$; encoding a graph with the given motifs; and
\item \textbf{Search Algorithm} (Sec. \ref{sec:algo}): 
Derive a subset selection algorithm for identifying a set of motifs to put in $\mt$.
\ceen
\end{problem}

\section{Encoding Schemes}
\label{sec:encoding}


MDL-encoding of a dataset with a model can be thought to involve a Sender and a Receiver communicating over a channel, where the Sender---who has the sole knowledge of the data---generates a bitstring based on which the
Receiver can reconstuct the original data on their end \textit{losslessly}.
To this end, the Sender first sends over the model, in our case the motif table $\mt$, which can be thought as a
``code-book'' that establishes a certain ``language'' between the Sender and the Receiver.
The Sender then encodes the data instances using the code-words in the ``code-book'', in our case the graph $G_j$'s using the motifs in the $\mt$.

It is important to note that existing works \cite{Cook94,bloem2020large} presented different graph/substructure encoding schemes, however, each has its own limitation and does not work in our settings. The encoding in \cite{Cook94} is lossy and may not reflect accurate code length, while the work in \cite{bloem2020large} is restricted on simple graphs with non-overlapping nodes of motifs' occurrences. Our encoding algorithm is both lossless and applicable on multi-graph with node-overlapping occurrences of motifs.


\subsection{Encoding the Motif Table}
\label{ssec:mt}

The motif table $\mt$ is simply a 
two-column translation table that has
motifs in the first column, and a unique code-word corresponding to each motif in the second column, as illustrated in {Fig. \ref{fig:demo}} (b).
We use $\m$ to refer to the set of motifs in $\mt$.
A motif, denoted by (lower-case) $g$, is a connected, directed, node-labeled, \textit{simple} graph, with possible self-loops on the nodes. For $g\in \m$, $code_{\mt}(g)$ (or $c$ for short) denotes its code-word.\footnote{To ensure \textit{unique} decoding, we assume
 \textit{prefix code(word)s}, in which no code is the prefix of another.}

To encode the motif table, the Sender encodes each individual motif $g_i$ in $\mt$ and also sends over the code-word $c_i$ that corresponds to $g_i$.
Afterwards, for encoding a graph with the $\mt$, every motif that finds an occurrence in the graph is simply communicated through its unique code-word only.

The specific materialization of the code-words (i.e., the bitstrings themselves) are not as important to us as their \textit{lengths}, which affect the total graph encoding length. Each code length $L(c_i)$ depends on the number of times that $g_i$ is \textit{used} in the encoding of the graphs in $\mG$, denoted $usage_{\mG}(g_i)$---intuitively, the more frequently a motif is used, the shorter its code-word is; so as to achieve compression (analogous to compressing text by assigning frequent words a short code-word). 
Formally, the optimal prefix code length for $g_i$ can be calculated through the Shannon entropy \cite{Rissanen78}:
\beq
\label{codelen}
L(c_i) \;=\; |code_{\mt}(g_i)| \;=\; -\log_2[P(g_i | \mG)]\;,
\eeq 

\noindent where $P$ is a probability distribution of $g_k \in \m$ for $\mG$:
\begin{align}
\label{prob}
P(g_i | \mG) & 
 = 
\frac{\sum_{G_j\in \mG} usage_{G_j}(g_i)}{\sum_{g_k\in \m} \sum_{G_j\in \mG} usage_{G_j}(g_k)}\;.
\end{align}

\begin{algorithm}[!t]
	\caption{{\sc Motif Encoding}}
		\begin{algorithmic}[1]
			\Require Motif $g_i=(V_i, E_i)$
			\Ensure Encoding of $g_i$ $\blacktriangleright$ Note: the values after symbol $\triangleright$ summed over the course of the algorithm provides the total encoding length $L(g_i)$
			\State Encode $n_i = |V_i|$; \# of nodes in $g_i$ $\;\; \triangleright\; L_{\N}(n_i)$\footnotemark
			
			\Repeat
			\State Pick an unmarked node $v\in V_i$ at random where $indeg(v) = 0$ (if none, pick any unmarked node), and mark $v$
			
			\Procedure{RecurseNode}{$mid(v)$}
			\State Encode $mid(v)$; $v$'s motif-node-ID $\;\; \triangleright\; \log_2(n_i)$			
			\State Encode $v$'s node label $\;\; \triangleright\; \log_2(T)$
			\State Encode \# of $v$'s out-neighbors $\mN_{\text{out}}(v)$ $ \triangleright\; L_{\N}(\text{outdeg}(v))$
			\State Encode motif-node-IDs of $\mN_{\text{out}}(v)$ $\;\; \triangleright\; \log_2{n_i\choose \text{outdeg}(v)}$
			
			\ForEach {unmarked node $u \in \mN_{\text{out}}(v)$}
			\State Mark $u$
			\State {\sc RecurseNode}$(mid(u))$
			\EndFor
			
			\EndProcedure
			
			\Until{all nodes in $V_i$ are marked}
		\end{algorithmic}
	\label{alg:motifencode}
\end{algorithm}

\footnotetext{MDL-optimal cost of integer $k$ is $L_{\N}(k) = \ls k + \log_2 c$; 
	$c \approx 2.865$; $\ls k = \log_2 k + \log_2(\log_2 k ) + \ldots$ summing only the positive terms \cite{Rissanen78}.}

We provide the details of how the motif usages are calculated in the next section, when we introduce graph encoding. 

Next, we present how a motif $g_i$ is encoded. Let $n_i$ denote the number of nodes it contains. (e.g., $g_1$ in {Fig. \ref{fig:demo}} (b) contains 3 nodes.) 
The encoding progresses recursively in a DFS (Depth First Search)-like fashion, as given in Algo. \ref{alg:motifencode}.
As noted, the encoding lengths summed over the course of the algorithm provides $L(g_i)$, which can be explicitly written as
follows:
\begin{align}
L(g_i) = & L_{\N}(n_i) + \sum_{v\in V_i} \log_2(n_i) + \log_2(T) \nonumber \\
 & + L_{\N}(\text{outdeg}(v)) + \log_2{n_i\choose \text{outdeg}(v)}
\end{align}

Overall, the total model encoding length is given as
\begin{tcolorbox}[ams align,myformula]
L(\mt) = L_{\N}(T) + \sum_{g_i\in \mt} \big[ L(g_i) + L(c_i) \big] \;,
\end{tcolorbox}
\noindent where we first encode the number of unique node labels, followed by the entries (motifs and codes) in the motif table.


\subsection{Encoding a Graph given the Motif Table}
\label{ssec:graph}

To encode a given graph $G_j$ based on a motif table $\mt$, we ``cover'' its edges by a set of motifs in the $\mt$.
To formally define {coverage}, we first introduce a few definitions.


\begin{definition}[Occurrence]
	An occurrence (a.k.a. a match) of a motif $g_i=(V_i,E_i)$ in a graph $G_j=(V_j, E_j)$ is a simple subgraph of $G_j$,
	denoted $o(G_j, g_i) \equiv g_{ij} = (V_{ij}, E_{ij})$ where $V_{ij}\subseteq V_j$ and $E_{ij}\subseteq E_j$,
	which is {\em isomorphic} to $g_i$,
	such that there exists a structure- and label-preserving bijection between node-sets $V_i$ and $V_{ij}$.
	Graph isomorphism is denoted $g_{ij} \simeq g_i$.
\end{definition}

Given a motif occurrence $g_{ij}$, we say that $g_{ij}$ \textit{covers} the edge set $E_{ij}\subseteq E_j$ of $G_j$.
The task of encoding a graph $G_j$ is then to \textit{cover all of its edges} $E_j$ using the motifs in $\mt$.  

\begin{definition}[Cover Set]
	Given a graph $G_j$ and a motif table $\mt$, $\coverset(G_j,\mt)$ is a set of motif occurrences so that
	\cbit
	\item If $g_{kj} \in \coverset(G_j,\mt)$, then $g_k\in \m$,
	\item $\bigcup_{g_{kj}\in \coverset(G_j,\mt)} E_{kj} = E_j$, and
	\item If $g_{kj}, g_{lj} \in \coverset(G_j,\mt)$, then $E_{kj} \cap E_{lj} = \emptyset$.	
	\ceit
\end{definition}
We say that $\coverset(G_j,\mt)$, or $\coverset_j$ for short, covers $G_j$.
The last item enforces the motif occurrences to cover \textit{non-overlapping edges} of a graph,
which in turn ensures that the coverage of a graph is always unambiguous.
This is mainly a computational choice---allowing overlaps would enable many possible covers, where enumerating and computing all of them would significantly increase the computational cost in the search of a motif table (See Eq. \eqref{obj}).

To encode a $G_j$ via $\mt$, the Sender
communicates the code of the motif associated with each occurrence in $G_j$'s cover set:
$$
G_j \; \rightarrow \; \{ code_{\mt}(g_i) \;|\; g_{ij} \in  \coverset(G_j,\mt)\} \;.
$$

However, it is not enough for the Receiver to simply know which collection of motifs occur in a graph in order to decode the graph \textit{losslessly}.
The Sender needs to encode further information regarding the {\em structure} of the graph. 
This can be achieved by, in addition to encoding {\em which} motif occurs, also encoding the specific \textit{graph-node-IDs} that correspond to or align with the motif-node-IDs.
The motif-node-IDs of a motif $g_i$ is simply the increasing set $\{1,\ldots,n_i\}$.
Then, the sequence of matching node IDs in graph $G_j$ is simply a permutation of $n_i$ unique values in $\{1,\ldots, |V_j|\}$, which can be encoded using $L_{\text{perm}}(|V_j|,n_i)$ bits, where
\beq
\label{perm}
L_{\text{perm}}(|V_j|,n_i) = \log_2(|V_j|\cdot|V_j-1|\cdot\ldots\cdot |V_j-n_i+1|) \;.
\eeq

Moreover, note that a motif can occur multiple times in a graph, possibly on the same set of node IDs (due to the input graphs being \textit{multi}-graphs) as well as on different sets.
We denote the different occurrences of the same motif $g_i$ in graph $G_j$'s cover set by $\mO(g_i, \coverset_j) = \langle g_{ij}^{(1)}, g_{ij}^{(2)}, \ldots, g_{ij}^{(c_{ij})} \rangle$ where $c_{ij}$ denotes the total count. 
See for example {Fig. \ref{fig:demo}} where motif $g_1$ shown in (b) exhibits three occurrences in the graph shown in (a), highlighted with dashed edges. 
Notice that two of those occurrences are on the same set of nodes, and the third one on a different set. (Note that occurrences can have different yet overlapping node sets, as in this example, as long as the edge sets are non-overlapping.)
We refer to the number of occurrences of a motif $g_i$ on the {\em same} set of nodes of a graph $G_j$, say $V\subseteq V_{j}$, as its {\em multiplicity on $V$}, denoted $m(g_i, G_j, V)$.
Formally,
\beq
m(g_i, G_j, V) = \big|\big\{ g_{ij}^{(k)} \;|\; V_{ij}^{(k)} =  V \big\} \big|\;, \text{where} \; V \subseteq V_j\;.
\eeq


Having provided all necessary definitions, Algo. \ref{alg:graphencode} presents how a graph $G_j$ is encoded.
For an occurrence $g_{kj}$ in its cover set, the corresponding motif $g_k$ is first communicated by simply sending over its code-word (line 2).
(Note that having encoded the motif table, the Receiver can translate the code-word to a specific motif structure.)
Then, the one-to-one correspondence between the occurrence nodes $V_{kj}$ and those of the motif $V_k$ is encoded (line 3).
Next, the multiplicity of the motif on the same set of nodes $V_{kj}$ is encoded (line 4) so as to cover as many edges of $G_j$ with few bits. 
Having been encoded, all those occurrences on $V_{kj}$ are then removed from $G_j$'s cover set (line 5).


Let us denote by $\overline{\coverset}_j \subseteq \coverset_j$ the \textit{largest} subset of occurrences in $\coverset_j$ with \textit{unique} node-sets, such that $V_{kj} \neq V_{lj}, \forall \{g_{kj}, g_{lk}\} \in \overline{\coverset}_j$.
Then, the overall graph encoding length for $G_j$ is given as

\begin{tcolorbox}[ams align,myformula]
	\label{anomaly_score}
	L(G_j|\mt) = \sum_{g_{kj} \in \overline{\coverset}_j} & 
	L(code_k) + L_{\text{perm}}(|V_j|,n_k)  \nonumber \\ 
	& + L_{\N}(m(g_k, G_j, V_{kj})) \;\;.
\end{tcolorbox}

The \textit{usage} count of a motif $g_i \in \m$ that is used to compute its code-word \textit{length} (Eq.s \eqref{codelen} and \eqref{prob}) is defined as
\beq
usage_{\mG}(g_i) = \sum_{G_j \in \mG} c_{ij} = \sum_{G_j \in \mG}\; \sum_{g_{kj} \in \coverset_j} \mathbbm{1}(g_{kj} \simeq g_i)
\eeq

\noindent where $\mathbbm{1}(g_{kj} \simeq g_i)$ returns 1 if $g_{kj}$ is both structure- and label-isomorphic to $g_i$, that is if $g_{kj}$ is an occurrence of $g_i$, and 0 otherwise.
The basic idea is to encode each motif \textit{only once} (as discussed in Sec. \ref{ssec:mt}), and then to encode occurrences of a motif in the data $\mG$ by simply providing a \textit{reference} to the motif, i.e., its code-word. This way we can achieve compression by assigning high-occurrence motifs a \textit{short} code-word, the principle behind Shannon's entropy.

\begin{algorithm}[!t]
	\caption{{\sc Graph Encoding}}
	\small{
		\begin{algorithmic}[1]
			\Require Graph $G_j$, Motif table $\mt$, $\coverset(G_j,\mt)$ 
			\Ensure Encoding of $G_j$ such that it can be decoded losslessly
			
			\Do
			\State Pick any occurrence $g_{kj} \in \coverset(G_j,\mt)$ and communicate $code_{\mt}(g_k)$ $\;\; \triangleright\; |code_{\mt}(g_k)| = L(code_k)$ \;\; Eq. \eqref{codelen}
			\State Encode matching graph-node-IDs 
			$\triangleright\; L_{\text{perm}}(|V_j|,n_k)$ \; Eq. \eqref{perm}
			
			\State Encode $g_k$'s {multiplicity} on $V_{kj}$
			$\;\; \triangleright\; L_{\N}(m(g_k, G_j, V_{kj}))$
			
			\State $\coverset(G_j,\mt) \leftarrow \coverset(G_j,\mt) \backslash \big\{ g_{lj} \;|\; V_{lj} =  V_{kj} \big\}$
			
			\doWhile{$\coverset(G_j,\mt) \neq \emptyset$}
			
		\end{algorithmic}
	}
	\label{alg:graphencode}
\end{algorithm}

\section{Search Algorithm}
\label{sec:algo}



Our aim is to compress as a large portion of the input graphs as possible using motifs. This goal can be restated as finding a large set of non-overlapping motif occurrences that cover these graphs.
We set up this problem as an instance of the Maximum Independent Set ($MIS$) problem on, what we call, the occurrence graph $\go$. 
In $\go$, the nodes represent motif occurrences and edges connect two occurrences that share a common edge. 
\mis ensures that occurrences in the solution are non-overlapping (thanks to independence, no two are incident to the same edge). Moreover,
it helps us identify motifs that have large usages, i.e. number of non-overlapping occurrences (thanks to maximality), which is associated with shorter code length and hence better compression.

In this section, first we describe how we set up and solve the MIS problems, which provides us with a set of candidate motifs that can go into the $\mt$ as well as their (non-overlapping) occurrences in input graphs. We then present a search procedure for selecting a subset of motifs among those candidates to minimize the total encoding length in Eq.~\eqref{obj}.

\subsection{Step 1: Identifying Candidate Motifs \& Their Occurrences}

As a first attempt, we explicitly construct a $\go$ per input graph and solve $MIS$ on it.
Later, we present an efficient way for solving $MIS$ without constructing any $\go$'s, which cuts down memory requirements drastically.

\subsubsection{\bf First Attempt: Constructing $\go$'s Explicitly}
$\;$\\
{\bf Occurrence Graph Construction.~} For each $G_j\in \mG$ we construct a $\go=(\vo,\eo)$ as follows. For $k=3,\ldots,10$, we enlist all \textit{connected} induced $k$-node subgraphs of $G_j$, each of which corresponds to an occurrence of some $k$-node motif in $G_j$. All those define the node set $\vo$ of $\go$.
If any two enlisted occurrences share at least one common edge in $G_j$, we connect their corresponding nodes in $\go$ with an edge.

Notice that we do not explicitly enumerate all possible $k$-node labeled motifs and then identify their occurrences in $G_j$, if any, which would be expensive due to subgraph enumeration (esp. with many node labels) and numerous graph isomorphism tests. The above procedure yields occurrences of all \textit{existing} motifs in $G_j$ \textit{implicitly}.

\vspace{0.05in}
\noindent
{\bf Greedy $\mis$ Solution.~} For the $\mis$ problem, we employ a greedy algorithm (Algo. \ref{alg:mis}) that sequentially selects the node with the minimum degree to include in the solution set $\mathcal{O}$. It then removes it along with its neighbors from the graph until $\go$ is empty. 
Let $deg_{\go}(v)$ denote the degree of node $v$ in $\go$ (the initial one or the one after removing nodes from it along the course of the algorithm), and
$\mN(v) = \{ u \in \vo | (u,v) \in \eo \}$ be the set of $v$'s neighbors in $\go$.

{\em Approximation ratio:~} The greedy algorithm for $MIS$ provides a $\Delta$-approximation \cite{Halldorsson97} where $\Delta = \max_{v \in \vo} deg_{\go}(v)$.
In our case $\Delta$ could be fairly large, e.g., in 1000s, as many occurrences overlap due to edge multiplicities. Here we strengthen this approximation ratio to $\min\{ \Delta, \Gamma \}$ as follows, where $\Gamma$ is around 10 in our data. 

\begin{theorem}
	For $MIS$ on occurrence graph $\go = (\vo,\eo)$, the greedy algorithm in Algo.~\ref{alg:mis} achieves an approximation ratio of $\min\{ \Delta, \Gamma \}$ where $\Delta = \max_{v \in \vo} deg_{\go}(v)$ and $\Gamma$ is the maximum number of edges in a motif that exists in the occurrence graph.
	\label{thm:mis_approx_ratio}
\end{theorem}

\begin{proof}
	Let $MIS(G_{\mathcal{O}})$ denote the number of nodes in the optimal solution set for $MIS$ on $G_{\mathcal{O}}$. We prove that
	\begin{align}
	|\mathcal{O}| \geq \frac{1}{\min\{\Delta, \Gamma\}} MIS(G_{\mathcal{O}}).
	\end{align}
	
	Let $G_{\mathcal{O}}(v)$ be the vertex-induced subgraph on $G_{\mathcal{O}}$ by the vertex set $\mathcal{N}_v = \{v\} \cup \mathcal{N}(v)$ when node $v$ is selected by the greedy algorithm to include in $\mathcal{O}$. $G_{\mathcal{O}}(v)$ has $\big\{(u,w) \in E_{\mathcal{O}} | u,w \in \mathcal{N}_v \big\}$ as the set of edges. In our greedy algorithm, when a node $v$ is selected to the solution set $\mathcal{O}$, $v$ and all of its current neighbors are removed from $G_{\mathcal{O}}$. Hence,
	\begin{align}
		\forall u, v \in \mathcal{O}, \mathcal{N}_u \cap \mathcal{N}_v = \emptyset. \nonumber
	\end{align}
	
	\noindent In other words, all the subgraphs $G_{\mathcal{O}}(v), v \in \mathcal{O}$ are independent (they do not share any nodes or edges). Moreover, since the greedy algorithm runs until $V_{\mathcal{O}} = \emptyset$ or all the nodes are removed from $\mathcal{O}$, we also have
	\begin{align}
		\cup_{v \in \mathcal{O}} \mathcal{N}_v = V_{\mathcal{O}}. \nonumber
	\end{align}
	
	Then, we can derive the following upper-bound for the optimal solution based on the subgraphs $G_{\mathcal{O}}(v), \forall v \in \mathcal{O}$:
	\begin{align}
		MIS(G_{\mathcal{O}}) \leq \sum_{v \in \mathcal{O}} MIS(G_{\mathcal{O}}(v)).
		\label{eq:mis_gr}
	\end{align}
	
	This upper-bound of $MIS(G_{\mathcal{O}})$ can easily be seen by the fact that the optimal solution set in $G_{\mathcal{O}}$ can be decomposed into subsets where each subset is an independent set for a subgraph $G_{\mathcal{O}}(v)$, hence, it is only a feasible solution of the $MIS$ problem on that subgraph with a larger optimal solution. 
	
	Next, we will find an upper-bound for $MIS(G_{\mathcal{O}}(v))$ and then plug it back in Eq.~\ref{eq:mis_gr} to derive the overall bound. There are two different upper-bounds for $MIS(G_{\mathcal{O}}(v))$ as follows:
	
	\textit{Upper-bound 1}: Assume that $E_{\mathcal{O}} \neq \emptyset$, otherwise the greedy algorithm finds the optimal solution which is $V_{\mathcal{O}}$ and the theorem automatically holds, we establish the first bound:
	\begin{align}
		MIS(G_{\mathcal{O}}(v)) & \leq \max\{1, |\mathcal{N}(v)|\} = \max \{1, deg_{G_{\mathcal{O}}}(v) \} \nonumber \\ & \leq \max_{u \in V} deg_{G_{\mathcal{O}}}(u) \leq \Delta
		\label{eq:delta}
	\end{align}
	Note that $\Delta$ is the maximum degree in the initial graph which is always larger or equal to the maximum degree at any point after removing nodes and edges.
	
	\textit{Upper-bound 2}: Let us consider the case that the optimal solution for $MIS$ on $G_{\mathcal{O}}(v)$ is the subset of $\mathcal{N}(v)$. Note that each node $v \in V_{\mathcal{O}}$ is associated with an occurrence which contains a set of edges in $G_j$, denoted by $E_j(v) = \{ e_{v1}, e_{v2}, \dots \}$, and, $u$ and $v$ are neighbors if $E_j(u) \cap E_j(v) \neq \emptyset$. Moreover, any pair of nodes $u, w$ in the optimal independent set solution of $G_{\mathcal{O}}(v)$ satisfies $E_j(u) \cap E_j(w) = \emptyset$. Thus, the largest possible independent set is $\{u_1, \dots, u_{|E_j(v)|}\}$ such that $E_j(u_l) \cap E_j(v) = \{e_{vl}\}$ ($u_l$ must be a minimal neighbor - sharing a single edge between their corresponding occurrences) and $E_j(u_l) \cap E_j(u_k) = \emptyset, \forall l \neq k$ (every node is independent of the others). Therefore, we derive the second upper-bound:
	\begin{align}
		MIS(G_{\mathcal{O}}(v)) \leq |\{u_1, \dots, u_{|E_j(v)|}\}| = |E_j(v)| \leq \Gamma.
		\label{eq:gamma}
	\end{align}
	
	Combining Eq.s~\ref{eq:delta} and \ref{eq:gamma} gives us a stronger upper-bound:
	\begin{align}
		MIS(G_{\mathcal{O}}(v)) \leq \min\{\Delta, \Gamma\}.
		\label{eq:ub}
	\end{align}
	
	Plugging Eq.~\ref{eq:ub} back to Eq.~\ref{eq:mis_gr}, we obtain
	\begin{align}
		& MIS(G_{\mathcal{O}}) \leq \sum_{v \in \mathcal{O}} \min\{\Delta, \Gamma\} = |\mathcal{O}| \times \min\{\Delta, \Gamma\}, \nonumber \\
		\;\;\Rightarrow &\;\; |\mathcal{O}| \geq \frac{1}{\min\{\Delta, \Gamma\}} MIS(G_{\mathcal{O}}). \nonumber
		\;\; \hspace{3cm}
	\end{align}
\end{proof}

\begin{algorithm}[!t]
	\caption{{\sc Greedy Algorithm for $\mis$}}
	\small{
		\begin{algorithmic}[1]
			\Require Motif occurrence graph $\go = (\vo,\eo)$ for graph $G_j$
			\Ensure Set $\mO \subset \vo$ of independent nodes (non-overlapping occurrences)
			
			\State Solution set $\mO = \emptyset$
			\While{$\vo \neq \emptyset$}
			\State $v_{\min} = \arg\min_{v \in \vo} deg_{\go}(v)$
			\State $\mathcal{O} := \mathcal{O} \cup \{v_{\min}\}$
			\State Remove $v_{\min}$ and $\mN(v_{\min})$ along with incident edges from $\vo$ and $\eo$ accordingly
			\EndWhile
			
		\end{algorithmic}
	}
	\label{alg:mis}
\end{algorithm}

{\em Complexity analysis:~} Algo.~\ref{alg:mis} requires finding the node with minimum degree (line~3) and removing nodes and incident edges (line~5), hence, a na\"ive implementation would have $O(|\vo|^2 + |\eo|)$ time complexity; $O(|\vo|^2)$ for searching nodes with minimum degree at most $|\vo|$ times (line~3) and $O(\eo)$ for updating node degrees (line~5). If we use a priority heap to maintain node degrees for a quick minimum search, time complexity becomes $O((|\vo| + |\eo|) \log |\vo|)$; $|\vo| \log |\vo|$ for constructing the heap initially and $|\eo| \log |\vo|$ for updating the heap every time a node and its neighbors are removed (includes deleting the degrees for the removed nodes and updating those of all $\mN(v)$'s neighbors).

Algo.~\ref{alg:mis} takes $O(|\vo| + |\eo|)$ to store the occurrence graph.

\subsubsection{\bf Memory-Efficient Solution: $\mis$ w/out Explicit $\go$}

Notice that the size of each $\go$ can be very large due to the combinatorial number of $k$-node subgraphs induced, which demands huge memory and time. Here we present an efficient version of the greedy algorithm that drastically cuts down the input size to MIS. Our new algorithm leverages a property stated as follows.

\begin{property}[Occurrence Degree Equality]
	\label{pro:uniform_deg}
	The nodes in the occurrence graph $\go$ of $G_j$ that correspond to occurrences that are enlisted based on subgraphs induced on the {\em same} node set in $G_j$ have exactly the {\em same} degree. 
\end{property}

Let us introduce a new definition called \emph{simple occurrence}. A simple occurrence $sg_{ij} = (V_{ij}, E_{ij})$ is a simple subgraph without the edge multiplicities of $G_j$ that is isomorphic to a motif. 
Let $\{sg_{1j}, sg_{2j}, \dots, sg_{tj}\}$ be the set of all simple occurrences in $G_j$.
Note that two simple occurrences may correspond to the same motif.

Recall that the greedy algorithm only requires the \textit{node degrees} in $\go$. Since all the nodes corresponding to occurrences that ``spring out'' of a simple occurrence have the same degree (Property \ref{pro:uniform_deg}), we simply use the simple occurrence as a ``compound node'' in place of all those degree-equivalent occurrences.\footnote{Given $sg_{ij}$, the number of its degree-equivalent occurrences in $G_j$ is the product of edge multiplicities $m(u,v)$ for $(u,v)\in E_{ij}$.}
As such, the nodes in $\go$ now correspond to simple occurrences \textit{only}.
The degree of each node (say, $sg_{ij}$) is calculated as follows.
\begin{align}
\label{eq:deg}
 \hspace{-0.1in}deg_{\go}(sg_{ij})  =&  \prod_{(u,v) \in E_{ij}}  m(u,v)   - \bigg[ \prod_{(u,v) \in E_{ij}} \big( m(u,v) - 1 \big) \bigg] - 1  \nonumber \\ 
& + \sum_{\{sg_{lj} | E_{lj} \cap E_{ij} \neq \emptyset\}}   \Big ( \prod_{(u,v) \in E_{lj} \backslash E_{ij}} m(u,v) \Big ) \;\; \times
 &\hspace{-0.1in}\Big( \prod_{(u,v) \in E_{lj} \cap E_{ij}} m(u,v) - \prod_{(u,v) \in E_{lj} \cap E_{ij}} \big(m(u,v) - 1\big) \Big), 
\end{align}
where $m(u,v)$ is the multiplicity of edge $(u,v)$ in $G_j$. 
The first line of Eq. \eqref{eq:deg} depicts the ``internal degree'' among the degree-equivalent occurrences that originate from $sg_{ij}$. The rest captures the ``external degree'' to other occurrences that have an overlapping edge.

\vspace{0.05in}
\noindent
{\bf Memory-Efficient Greedy $\mis$ Solution.~} The detailed steps of our memory-efficient greedy $\mis$ algorithm is given in Algo. \ref{alg:mem_greedy_mis}.
We first calculate degrees of all simple occurrences (line 2) and then sequentially select the one with the minimum degree, denoted $sg_{i^*j}$, included in the solution list $\mSO$ (lines 5, 6).
To account for this selection, we decrease multiplicities of all its edges $(u,v)\in E_{i^{*}j}$ by 1 (line 7) and recalculate the degrees of simple occurrences that overlap with $sg_{i^{*}j}$ (lines 8-12). If one of those simple occurrences that has an intersecting edge set with that of $sg_{i^{*}j}$ contains at least one edge with multiplicity equal to 0 (due to decreasing edge multiplicities in line 7), a special value of $deg_{\max}+1$ (line 3) is assigned as its degree. This is to signify that this compound node contains no more occurrences, and it is not to be considered in subsequent iterations (line 4).


\begin{algorithm}[!t]
	\caption{{\sc Memory-efficient Greedy  $\mis$}} 
	\small{
		\begin{algorithmic}[1]
			\Require Simple occurrences $sg_{1j}, \dots, sg_{tj}$, edge multiplicities in $G_j$
			\Ensure A list $\mSO$ of simple occurrences
			
			\State $\mSO = \emptyset$
			\State		Calculate degrees of all simple occurrences $deg_{\go}(sg_{ij}), \forall i = 1\ldots t$ using Eq.~\eqref{eq:deg}
			\State			$deg_{\max} = \max_{i = 1\ldots t} deg_{\go}(sg_{ij})$
			\While{$\exists i\in \{1\ldots t\} \;\; \text{s.t.}\;\; deg_{\go}(sg_{ij}) < deg_{\max} + 1$}
			\State		$i^* = \arg\min_{i \in\{ 1\ldots t\}} deg_{\go}(sg_{ij})$
			\State		$\mSO = \mSO \cup \{sg_{i^*j}\}$
			\State		$m((u,v)) = m((u,v)) - 1\;,\; \forall (u,v) \in E_{i^*j}$
			\For{$l = 1\ldots t, E_{lj} \cap E_{i^*j} \neq \emptyset$}
			\If{$m((u,v)) > 0 \; \forall (u,v) \in E_{lj}$}
			\State	Recalculate $deg_{\go}(sg_{lj})$ using Eq.~\eqref{eq:deg}
			\Else
			\State	$deg_{\go}(sg_{lj}) = deg_{\max} + 1$
			\EndIf
			\EndFor
			\EndWhile
		\end{algorithmic}
	}
	\label{alg:mem_greedy_mis}
\end{algorithm}

Notice that we need not even construct an occurrence graph in Algo. \ref{alg:mem_greedy_mis}, which directly operates on the set of simple occurrences in $G_j$, computing and updating degrees based on Eq. ~\eqref{eq:deg}.
Note that the same simple occurrence could be picked more than once by the algorithm. The number of times a simple occurrence appears in $\mSO$ is exactly the number of non-overlapping occurrences that spring out of it and get selected by Algo. \ref{alg:mis} on $\go$. As such $\mO$ and $\mSO$ have the same cardinality and each motif has the same number of occurrences and simple occurrences in $\mO$ and $\mSO$, respectively.
As we need the number of times each motif is used in the cover set of a graph (i.e., its usage), both solutions are equivalent.

\begin{figure}[!h]
	\centering
	\includegraphics[width = 0.55\linewidth]{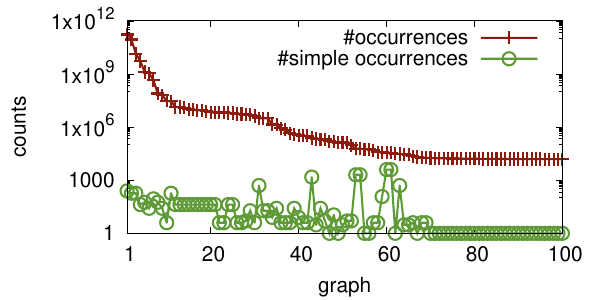} 
	\caption{Number of occurrences ($|\vo|$) and number of simple occurrences ($t$) for top 100 graphs in \sh{} database with highest number of occurrences (in decreasing order). Note that the number of simple occurrences is up to 9 orders of magnitude smaller, thus, our memory-efficient algorithm drastically reduces  the memory requirement by the same factor.}
	\label{fig:occ_counts1}
\end{figure}

{\em Complexity analysis:~}
Calculating $deg_{\go}(sg_{ij})$ of a simple occurrence takes $O(t \cdot \Gamma)$, as Eq.~\eqref{eq:deg} requires all $t$ intersecting simple occurrences $\{sg_{lj} \;\vert\; E_{lj} \cap E_{ij} \neq \emptyset\}$ where intersection can be done in $O(\Gamma)$, the maximum number of edges in a motif. Thus, Algo.~\ref{alg:mem_greedy_mis} requires $O(t^2 \cdot \Gamma)$ for line~2 to calculate all degrees.
Within the while loop (lines 4-12), the total number of degree recalculations (line 10) is bounded by 
$O(t \cdot \Gamma \cdot \max_{(u,v) \in E_j} m((u,v)))$ since each simple occurrence gets recalculated for at most $\Gamma \cdot \max_{(u,v) \in E_j} m((u,v))$ times. Finding the intersecting simple occurrences (line 8) is $(t \cdot \Gamma)$. Overall, Algo.~\ref{alg:mem_greedy_mis} time complexity is $O(t^2 \cdot \Gamma^2 \cdot \max_{(u,v) \in E_j} m(u,v))$.

The space complexity of Algo.~\ref{alg:mem_greedy_mis} is $O(t \cdot \Gamma + |E_j|)$, for $t$ simple occurrences with at most $\Gamma$ edges each, plus input edge multiplicities. Note that this is significantly smaller than the space complexity of Algo.~\ref{alg:mis}, i.e. $O(|\vo| + |\eo|)$, since $t \ll |\vo|$ and $|E_j| \ll |\eo|$. Our empirical measurements on our \sh{} dataset (see Table~\ref{tab:data} in Sec.~\ref{sec:experiment} for details) in Fig.~\ref{fig:occ_counts1} show that $|\vo|$ is up to 9 orders of magnitude larger than $t$.

\subsubsection{\bf Weighted Maximum Independent Set ($\wmis$)}
\vspace{0.05in}
$\;$
A motivation leading to our $\mis$ formulation is to cover \textit{as much of each input graph as possible} using motifs. 
The amount of coverage by an occurrence of a motif can be translated into the number of edges it contains.
This suggests a \textit{weighted} version of the maximum independent set problem, denoted $\wmis$, where we set the weight of an occurrence, i.e., node in $\go$, to be the number of edges in the motif it corresponds to. Hence, the goal is to maximize the total weight of the non-overlapping occurrences in the solution set.

\vspace{0.05in}
\noindent
{\bf Greedy $\wmis$ Solution.~}  For the weighted version of $\mis$, we also have a greedy algorithm with the \textit{same} approximation ratio as the unweighted one. The only difference from Algo.~\ref{alg:mis} is the selection of node $v_{\min}$ to remove (line 3). Let $w_v$ denote the weight of node $v$ in $\go$. Then,
\begin{align}
	v_{\min} = \arg\max_{v \in \vo} \frac{w_v}{\min\{w_v, \deg_{\go}(v)\} \cdot \max_{u \in \mN(v)} w_u}
	\label{eq:v_min_wmis}
\end{align}
Intuitively, Eq.~\eqref{eq:v_min_wmis} prefers selecting large-weight nodes (to maximize total weight), but those that do not have too many large-weight neighbors in the $\go$, which cannot be selected.

\vspace{0.1in}
\begin{theorem}
For $\wmis$ on weighted occurrence graph $\go = (\vo,\eo, W)$ where $w_v \in W$ is the weight of node $v \in \vo$, the greedy algorithm in Algo.~\ref{alg:mis} with node selection criterion in Eq.~\eqref{eq:v_min_wmis} achieves an approx. ratio of $\;\min\{ \Delta, \Gamma \}$ where $\Gamma = \max_{v \in \vo} w_v$ is the maximum number of edges in a motif.
\label{theo:wmis_approx_ratio}
\end{theorem}

\begin{proof}
	The proof is similar to that of Theorem~\ref{thm:mis_approx_ratio} with the key realization that the total weight of the optimal solution set of $WMIS$ on subgraph $G_{\mathcal{O}}(v)$ is upper-bounded by $\min\{w_v, \deg_{G_{\mathcal{O}}}(v)\} \cdot \max_{u \in \mathcal{N}(v)} w_u$ and when a node $v$ is selected from $G_{\mathcal{O}}$ by the greedy algorithm, due to the selection condition, we have,
	\begin{align}
	& \frac{w_v}{\min\{w_v, \deg_{G_{\mathcal{O}}}(v)\} \cdot \max_{u \in \mathcal{N}(v)} w_u} \geq \frac{w_{v^*}}{\min\{w_{v^*}, \deg_{G_{\mathcal{O}}}(v^*)\} \cdot \max_{u \in N(v^*)} w_u} \nonumber \\
	& \geq \frac{1}{\min\{w_{v^*}, \deg_{G_{\mathcal{O}}}(v^*)\}} \geq \frac{1}{\min \{\Delta, \Gamma\}},
	\end{align}
	
	\noindent where $v^* = \max_{u \in V_{\mathcal{O}}} w_u$, then $\forall u \in \mathcal{N}(v^*), w_{v^*} \geq w_u$. $\qed$
\end{proof}

\textbf{Memory-efficient Greedy $\wmis$ Solution.~} Similar to the unweighted case, we can derive a memory-efficient greedy algorithm for $\wmis$ since Property~\ref{pro:uniform_deg} holds for both degree $deg_{\go}(v)$ and $w_v$---the two core components in selecting nodes (Eq.~\eqref{eq:v_min_wmis}) in greedy algorithm for $\wmis$---since all occurrences of the same simple occurrence have the same number of edges.

\begin{algorithm}[!t]
	\caption{{\sc Motif Table Search}}
	\small{
		\begin{algorithmic}[1]
			\Require Database $\mG = \{G_1, \dots, G_J\}$, candidate motifs $\mC$, node labels $\mT$,
			$(W)MIS$ solution $\mSO_j \; \forall G_j\in  \mG$
			\Ensure $\mt$ containing set of motifs $\mM$

			\State $\text{Cnt}_j(s,d) = 0 \;\;, \text{for } j=1\ldots J, \text{and } \forall s,d \in \mT$
			
			\For{$G_j=(V_j,E_j)\in \mG$}
			\ForEach{$(u,v) \in E_j$} 
			\State $\text{Cnt}_j(t(u),t(v)) := \text{Cnt}_j(t(u),t(v)) + m(u,v)$
			\EndFor	
			\EndFor		
			
			\State $\mM = \{(s, d) | s,d \in \mT \text{ and } \exists j \in [1,J] \text{ where } \text{Cnt}_j(s,d) > 0\}$
			
			\While{$\mC \neq \emptyset$}			
			\For{$g=(V,E)\in \mC$}
			\For{$G_j\in \mG$}
			\State $\mO(g,\coverset_j) = \{sg \simeq g | sg \in \mSO_j \}$
			\State $\text{Cnt}_{jg} := \text{Cnt}_j$
			\ForEach{$(u,v) \in E$} 
			\State $\text{Cnt}_{jg}(t(u),t(v))$:=$\text{Cnt}_{jg}(t(u),t(v))$$-$$|\mO(g,\coverset_j)|$ 	
			\EndFor	
			\EndFor
			\State $\mM_{g} = \mM \cup \{g\}  \backslash 
			\{(s, d) | s,d \in \mT, 
			\text{Cnt}_{jg}(s,d) = 0 \;\forall j\}$	
			
			\State Get $L(\mM_{g}, \mG)$ using $\mO(g_k,\coverset_j)\; \forall g_k\in \mM_{g}$ s.t. $n_k\geq 3$ and
			$usage_{G_j}(s,d) := \text{Cnt}_{jg}(s,d)\; \forall(s,d)\in \mM_{g}$; $\forall j$
			
			\EndFor	
			
			\State $g_{\min} = \arg\min_{g \in \mC} L(\mM_g, \mathcal{G})$
			\If{$L(\mM_{g_{\min}}, \mathcal{G}) \leq L(\mathcal{M}, \mathcal{G})$}
			\State $\mM := \mM_{g_{\min}}\;$, 
			$\;\mC:=\mC\backslash \{g_{\min}\}\;$, 
			$\;\text{Cnt}_j:= \text{Cnt}_{jg_{\min}}$
			\Else
			\State {\bf break}
			\EndIf
			
			\EndWhile
		\end{algorithmic}
	}
	\label{alg:search}
\end{algorithm}

\subsection{Step 2: Building the Motif Table}

The $(W)MIS$ solutions, $\mSO_j$'s, provide us with non-overlapping occurrences of $k$-node motifs ($k\geq 3$) in each $G_j$.
The next task is to identify the subset of those motifs to include in our motif table $\mt$ so as to minimize the total encoding length in Eq. \eqref{obj}.
We first define the set $\mC$ of {\em candidate motifs}:
\beq
\mC = \{ g\simeq sg \;|\; sg \in \bigcup_{j} \mSO_j \}.
\eeq

We start with encoding the graphs in $\mG$ using the simplest code table that contains only the $2$-node motifs.
This code table, with optimal code lengths for database $\mG$, is called the {\em Standard Code Table}, denoted by $\smt$.
It provides the optimal encoding of $\mG$ when nothing more is known than the frequencies of labeled edges (equal to $usage$s of the corresponding $2$-node motifs), which are assumed to be fully independent.
As such, $\smt$ does not yield a good compression of the data but provides a practical bound.

To find a better code table, we use the best-first greedy strategy: Starting with $\mt:=\smt$, we try adding each of the candidate motifs in $\mC$ one at a time. Then, we pick the `best' one that leads to the largest reduction in the total encoding length. We repeat this process with the remaining candidates until no addition leads to a better compression or all candidates are included in the $\mt$, in which case the algorithm terminates. 

The details are given in Algo. \ref{alg:search}.
We first calculate the $usage$ of $2$-node motifs per $G_j$ (lines 1-4), and set up the $\smt$ accordingly (line 5).
For each candidate motif $g\in \mC$ (line 7) and each $G_j$ (line 8) we can identify the occurrences of $g$ in $G_j$'s cover set, $\mO(g,\coverset_j)$, which is equivalent to the simple occurrences selected by Algo. \ref{alg:mem_greedy_mis} that are isomorphic to $g$ (line 9).
When we insert a $g$ into $\mt$, $usage$ of some $2$-node motifs, specifically those that correspond to the labeled edges of $g$, decreases by $|\mO(g,\coverset_j)|$; the $usage$ of $g$ in $G_j$'s encoding (lines 10-12).
Note that the $usage$ of $(k\geq 3)$-node motifs already in the $\mt$ do not get affected,
since their occurrences in each $\mSO_j$, which we use to cover $G_j$, do \textit{not} overlap; i.e., their uses in covering a graph are independent.
As such, updating $usage$s when we insert a new motif to $\mt$ is quite efficient.
Having inserted $g$ and updated $usage$s, we remove $2$-node motifs that reduce to zero $usage$ from $\mt$ (line 13), and compute the total encoding length with the resulting $\mt$ (lines 14-15).
The rest of the algorithm (lines 16-20) picks the `best' $g$ to insert that leads to the largest savings, if any, or otherwise quits. 


\section{Experiments}
\label{sec:experiment}
%

{\bf Datasets.} Our work is initiated by a collaboration with industry, and \method is evaluated on large real-world datasets containing 
all transactions of 2016 (tens to hundreds of thousands transaction graphs) 
from 3 different companies, anonymously \sh, \hw, and \kd (proprietary) summarized in Table~\ref{tab:data}. These do not come with ground truth 
 anomalies.
For quantitative evaluation, our expert collaborators 
inject two types of anomalies into each dataset based on domain knowledge (Sec. \ref{ssec:performance}), and also qualitatively verify 
the detected anomalies
 from an accounting perspective (Sec. \ref{ssec:cs}).
 
\begin{figure}
	\centering
	\begin{subfigure}[t]{0.235\textwidth}
		\includegraphics[width=\linewidth]{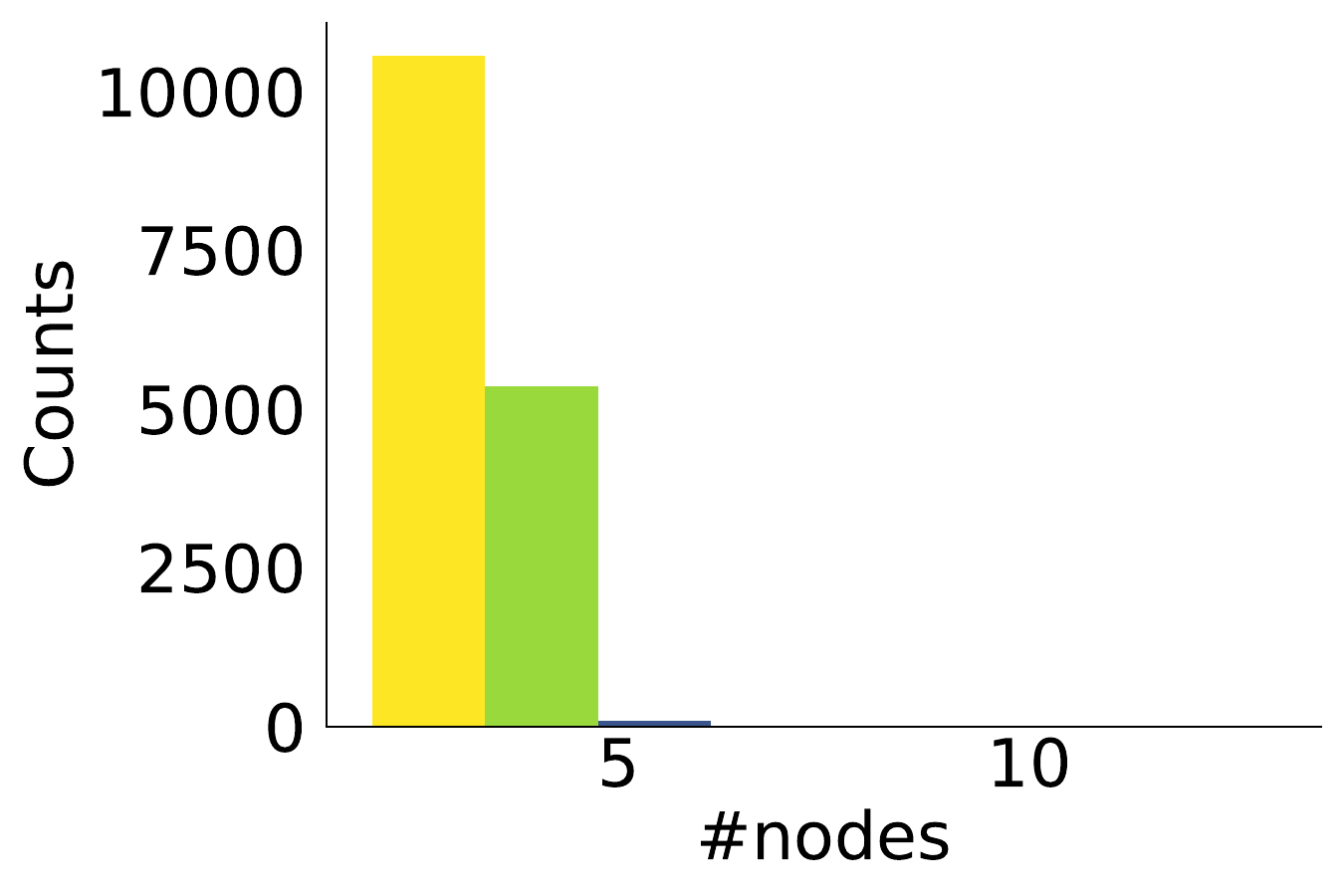}
		\caption{Distribution of the number of nodes.}
		\label{fig:sh_nodes}
	\end{subfigure}
	\hfill
	\begin{subfigure}[t]{0.235\textwidth}
		\includegraphics[width=\linewidth]{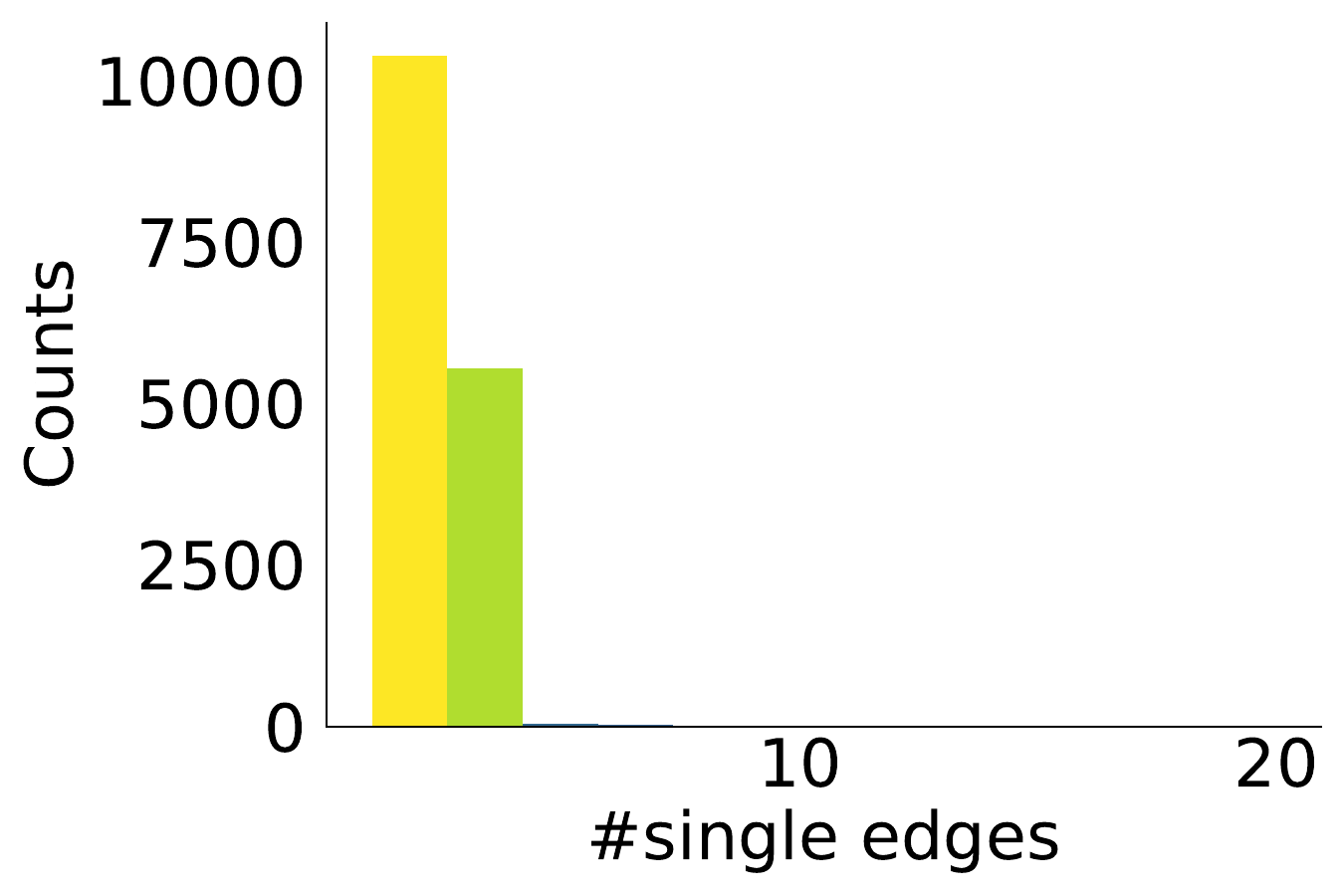}
		\caption{Distribution of the number of single edges.}
		\label{fig:sh_edges}
	\end{subfigure}
	\hfill
	\begin{subfigure}[t]{0.235\textwidth}
		\includegraphics[width=\linewidth]{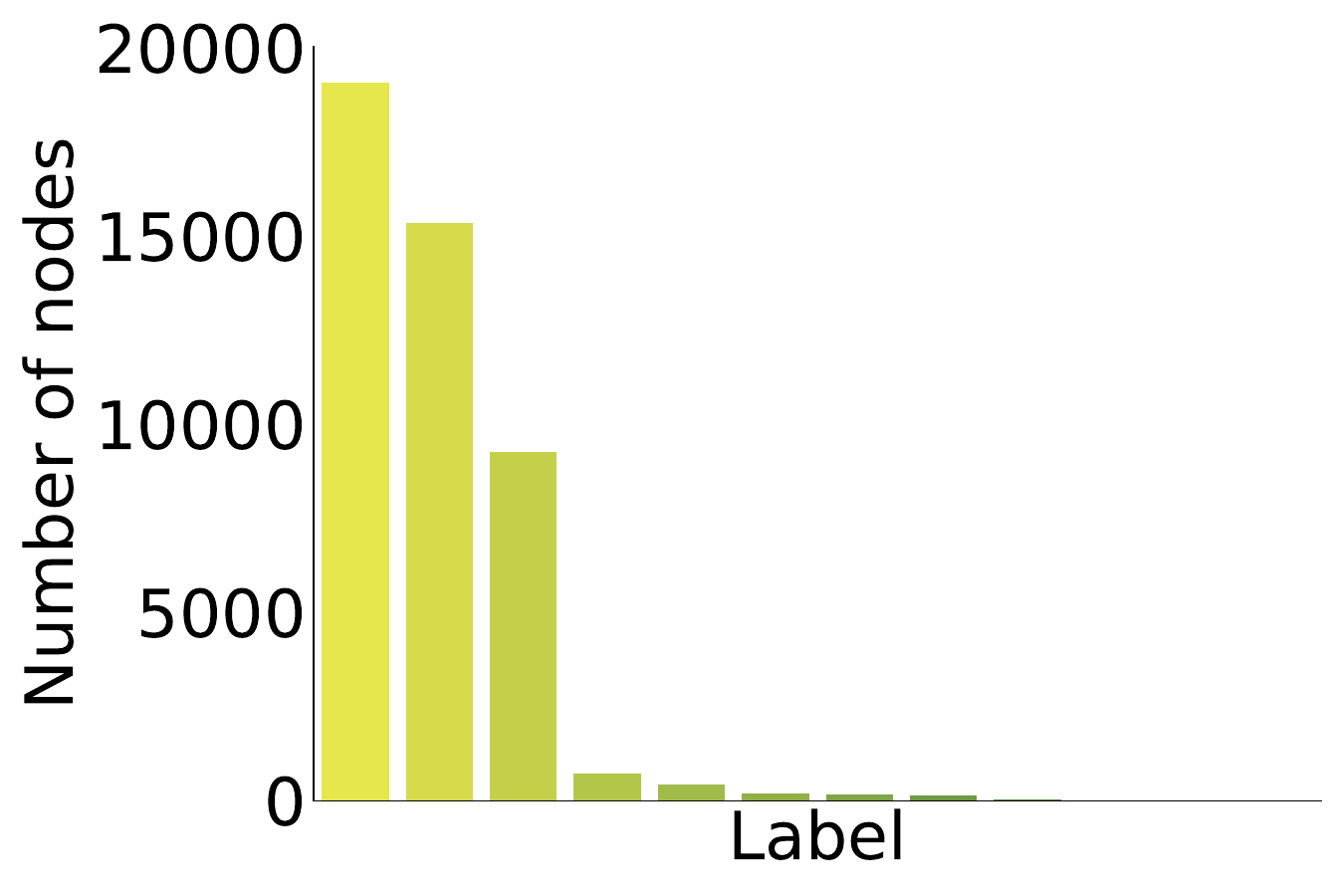}
		\caption{Distribution of node labels.}
		\label{fig:sh_node_type}
	\end{subfigure}
	\hfill
	\begin{subfigure}[t]{0.235\textwidth}
		\includegraphics[width=\linewidth]{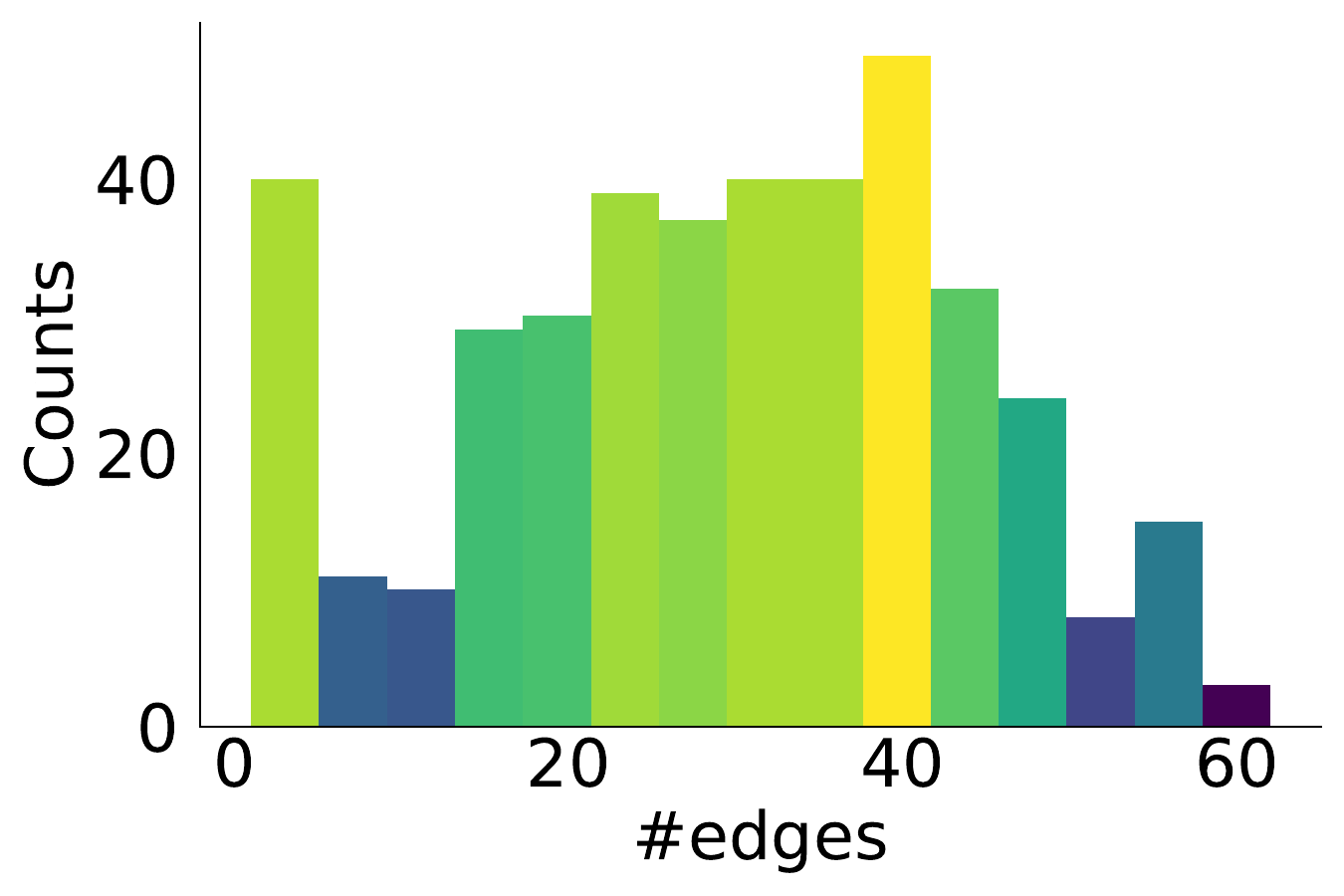}
		\caption{Distribution of multiplicities for a pair of node labels.}
		\label{fig:sh_edge_type}
	\end{subfigure}

	\caption{Statistical summary of graph characteristics in \sh{} database}
\end{figure}
 
Since our transaction data are not shareable, to facilitate reproducibility of the results, we generate synthetic data that resemble the real-world datasets used. In particular, we add \shsyn{} that contains random graphs generated based on statistical characteristics of \sh{} as follows: 1) A graph in \shsyn{} has the number of nodes randomly selected following its distribution in \sh{} depicted in Figure~\ref{fig:sh_nodes}; 2) The number of single edges in a graph is drawn randomly following the distribution in \sh{} illustrated in Figure~\ref{fig:sh_edges} (If the graph has more edges than the maximum number - $n\times(n-1)/2$, we restart the process); 3) For each node, we randomly assign its label following the label distribution in \sh{} as shown in Figure~\ref{fig:sh_node_type}; 4) For each single edge, we then randomly generate its multiplicity based on the distribution of the corresponding node label pair (an example is given in Figure~\ref{fig:sh_edge_type}). At the end, we have the same number of graphs with similar characteristics as \sh{} and share this data along with our code.

We also study the public Enron database \cite{hooi2017graph}, consisting of daily email graphs of its 151 employees over 3 years  surrounding the financial scandal. Nodes depict employee email addresses and edges indicate email exchanges. Each node is labeled with the employee's department (Energy Operations, Regulatory and Government Affairs, etc.) and edge multiplicity denotes the number of emails exchanged. 

\begin{table}[!h]
	\centering
	\vspace{-0.05in}
	\caption{Summary statistics of graph datasets used in experiments.}
	\label{tab:data}
	\begin{tabular}{lrccr}
		\toprule
		\textbf{Name} &\#Graphs & \#Node labels  & \#Nodes  
		& \#Multi-edges \\ 
		\midrule
		\sh & 38,122 & 11 & [2, 25] 
		& [1, 782] \\
		\shsyn & 38,122 & 11 & [2, 25] 
		& [1, 782] \\
		\hw & 90,274 & 11 & [2, 25] 
		& [1, 897]\\
		\kd & 152,105 & 10 & [2, 91] 
		& [1, 1774]\\
		\enron & 1,139	& 16 & [2, 87] & [1, 1356]\\	
		\bottomrule
	\end{tabular}
	\vspace{-0.15in}
\end{table}

\hide{
\subsection{Algorithm Analysis}

\subsubsection{\bf Compression Rate}
We first compare the total compression length using $SMT$ and the final $MT$ produced by \method in Fig.~\ref{fig:compress_rate} for graphs with at least 3 nodes.
\fgr[Compression rate of the final $MT$ to $SMT$. \label{fig:compress_rate}]{width = 0.85\linewidth}{FIG/compress_rate.pdf}
Compared to $\smt$, the compression rate of the final $\mt$ is around 4\% on \sh dataset due to the fact that most graphs are very small, e.g., only 150 having at least 5 nodes. On the other hand, the compression rates on \hw and \kd are much higher, >10\%, since they have many more large graphs, e.g., a few thousands with at least 5 nodes.

\subsubsection{\bf Efficiency}
$\;$\\
We evaluate the efficiency of our memory-efficient greedy algorithm (Alg.~\ref{alg:mem_greedy_mis}) compared to the standard greedy (Alg.~\ref{alg:mis}) by measuring the \#simple occurrences used in Alg.~\ref{alg:mem_greedy_mis} and \#occurrences required for Alg.~\ref{alg:mis} on \sh dataset in Fig.~\ref{fig:occ_counts}.
\fgr[\#occurrences and \#simple occurrences for top-100 graphs with highest \#occurrences. \label{fig:occ_counts}]{width = 0.85\linewidth}{FIG/occ_counts.pdf}
We see that the \#occurrences are up to 9 orders of magnitude larger than \#simple occurrences. In fact, we were not able to run Alg.~\ref{alg:mis} due to the humongous \#occurrences needed to be enumerated for the first few graphs in Fig.~\ref{fig:occ_counts}. In the contrary, \#simple occurrences are just in order of thousands and can be managed easily regarding memory and time.


\subsection{Motif Analysis}

All motifs in all three companies capture operating activities---such as Operating Revenues/Sales, Operating Expenses, Short- and Long-term Operating Assets)---consistent with the general accounting bookkeeping practices in that journal entries capture day-to-day company activities that tend to be repetitive over time and across different customers.

The key differences among motifs from different companies are about (1) the types and (2) the number of operating accounts. For example, LOA only appear in HW and ORS do not appear in KD. Separately, while OE appear in motifs from all three datasets, the number of distinct OE nodes is different across the three. 

These differences can be indicative of the different nature of these companies such as the line of business (merchandising vs. manufacturing, retail vs. wholesale businesses). For example, it is quite common for manufacturing companies to generate many journal entries involving LOA due to the heavy use of many different machinery, leading to repetitive depreciation bookings---linking LOA to OE, as in HW.  
At the same time, compared to retail companies, the number of distinct sales journal entries may be smaller for wholesale companies.

Finally, it should be noted that different motifs across companies can also be a result of different accounting policies/procedures deployed by different management teams, who choose to separate or combine several aspects 
when booking the activities.

\reminder{graph-visualization of motifs in $\mt$ (along with usages) and domain semantics}
}

\vspace{-0.05in}
\subsection{Anomaly Detection Performance}
\label{ssec:performance}

We show that \method is substantially better in detecting graph anomalies as compared to a list of baselines across various performance measures. 
The anomalies are injected by domain experts, which mimic schemes related to money laundering, entry error or malfeasance in accounting, specifically: 

\begin{compactitem} 
\setlength{\itemindent}{-.125in}
\item \emph{Path injection (money-laundering-like)}: ($i$) Delete a random edge $(u,v) \in E_j$, and ($ii$) Add a length- 2 or 3 path $u$--$w$(--$z$)--$v$ where {at least one edge} of the path is rare (i.e., exists in only a few $G_j$'s).
The scheme mimics money-laundering, where money is transferred through multiple hops rather than directly from the source to the target.

\item \emph{Label injection (entry-error or malfeasance)}: ($i$) Pick a random node $u \in V_j$, and ($ii$) Replace its label $t(u)$ with a random label $t \neq t(u)$. This scheme mimics either simply an entry error (wrong account), or malfeasance that aims to reflect larger balance on certain types of account (e.g., revenue) in order to deceive  stakeholders. 
\end{compactitem}

For path injection, we choose 3\% of 
graphs and inject anomalous paths that replace 10\% of edges (or 1 edge if 10\% of edges is less than 1). For label injection, we also choose 3\% of 
graphs and label-perturb 10\% of the nodes (or 1 node if 10\% of nodes is less than 1). We also tested with different severity levels of injection, i.e., 30\% and 50\% of edges or nodes, and observed similar results to those with 10\%. The goal is to detect those graphs with injected paths or labels.

\begin{table}[!h]
	\centering
	\caption{Detection performance of \textit{path anomalies} in various datasets. Numbers in \textbf{bold} highlight best results in the columns and \underline{underlined} numbers refer to the runner-up. In all performance measures, \method{} achieves significantly higher results than the others and maintains a large gap to the runner-up, which is not stable but varies depending on dataset and performance measure.}
	\begin{subtable}{\textwidth}
		\centering
		\caption{\sh dataset}
		\label{tab:anomaly_1_sh}
		\begin{tabular}{lrrrrr}
			\toprule
			\textbf{Method} & Prec@10 & Prec@100 & Prec@1000 & AUC & AP\\
			\midrule
			\method & \textbf{1.000} & \textbf{0.920} & \textbf{0.386} & \textbf{0.958} & \textbf{0.548} \\
			\smt & 0.100 & 0.280 & 0.352 & \underline{0.932} & \underline{0.413} \\
			\glocalkd & 0.400 & 0.252  & 0.331 & 0.916 & 0.405 \\
			\gbad & 0.200 & 0.495 & \underline{0.356} & 0.906 & 0.373 \\
			GF+\iforest & 0.100 & 0.120 & 0.237 & 0.926 & 0.210\\
			G2V+\iforest{} & \underline{0.800} & {0.750} & 0.308 & 0.886 & 0.383\\
			DGK+\iforest{} & 0.100 & 0.030 & 0.025 & 0.712 & 0.050\\
			\entropy & 0.100 & \underline{0.800} & 0.219 & 0.821 & 0.347\\
			\multi & 0.000 & 0.040 & 0.027 & 0.643 & 0.049\\
			\bottomrule
		\end{tabular}
		\vspace{0.2in}
	\end{subtable}
	\begin{subtable}{\textwidth}
		\centering
		\caption{\hw dataset}
		\label{tab:anomaly_1_hw}
		\begin{tabular}{lrrrrr}
			\toprule
			\textbf{Method} & Prec@10 & Prec@100 & Prec@1000 & AUC & AP\\
			\midrule
			\method & \textbf{0.900} & \textbf{0.990} & \textbf{0.999} & \textbf{0.995} & \textbf{0.772} \\
			\smt & 0.600 & 0.440 & 0.784 & 0.906 & \underline{0.733} \\
			\glocalkd & 0.000 & 0.269 & 0.531 & 0.843 & 0.563 \\
			\gbad & \underline{0.800} & 0.710 & 0.685 & 0.930 & 0.555 \\
			GF+\iforest & 0.400 & 0.230 & 0.497 & 0.959 & 0.429 \\
			G2V+\iforest{} & 0.000 & 0.100 & 0.819 & 0.824 & 0.380\\
			DGK+\iforest{} & 0.300 & 0.140 & 0.023 & 0.858 & 0.097\\
			\entropy & 0.300 & \underline{0.820} & \underline{0.896} & \underline{0.981} & 0.571 \\
			\multi & 0.000 & 0.020 & 0.029 & 0.719 & 0.106 \\
			\bottomrule
		\end{tabular}
		\vspace{0.2in}
	\end{subtable}
	\begin{subtable}{\textwidth}
		\centering
		\caption{\shsyn dataset}
		\label{tab:anomaly_1_shsyn}
		\begin{tabular}{lrrrrr}
			\toprule
			\textbf{Method} & Prec@10 & Prec@100 & Prec@1000 & AUC & AP\\
			\midrule
			\method & \textbf{0.600} & \textbf{0.860} & \textbf{0.412} & \textbf{0.982} & \textbf{0.664} \\
			\smt & 0.300 & 0.440 & 0.212 & 0.911 & 0.398 \\
			\glocalkd & 0.300 & 0.495 & 0.258 & 0.932 & 0.473 \\
			\gbad & 0.300 & 0.523 & \underline{0.318} & 0.937 & \underline{0.468} \\
			GF+\iforest & 0.200 & 0.420 & 0.252 & 0.894 & 0.340\\
			G2V+\iforest{} & \underline{0.500} & \underline{0.640} & 0.310 & \underline{0.943} & 0.450\\
			DGK+\iforest{} & 0.100 & 0.030 & 0.025 & 0.712 & 0.050\\
			\entropy & 0.300 & 0.540 & 0.294 & 0.921 & 0.433\\
			\multi & 0.000 & 0.050 & 0.036 & 0.662 & 0.102\\
			\bottomrule
		\end{tabular}
	\end{subtable}
\end{table}


\begin{table}[!h]
	\centering
	\caption{Detection performance of \textit{label anomalies} in various datasets (\textbf{bold} and \underline{underlined} numbers refer to the best and runner-up results). We observe notably higher performance of \method{} compared to the competing baselines, where it ranks either at the top or is the runner-up across the board. In (b), $-$ depicts cases for \gbad{} and G2V+\iforest{}, which failed to complete within 5 days.}
	\begin{subtable}{\textwidth}
		\centering
		\caption{\hw dataset}
		\label{tab:anomaly_2_hw}
		\begin{tabular}{lrrrrr}
			\toprule
			\textbf{Method} & Prec@10 & Prec@100 & Prec@1000 & AUC & AP\\
			\midrule
			\method & \textbf{0.800} & \textbf{0.720} & \textbf{0.709} & \underline{0.918} & \underline{0.359} \\
			\smt & 0.000 & 0.100 & 0.174 & 0.883 & 0.192 \\
			\glocalkd & 0.600 & 0.640 & 0.663 & 0.902 & 0.301 \\
			\gbad & \underline{\textbf{0.800}} & \underline{0.710} & \underline{0.685} & \textbf{0.920} & \textbf{0.555} \\
			GF+\iforest & 0.200 & 0.080 & 0.027 & 0.832 & 0.092 \\	
			G2V+\iforest{} & 0.000 & 0.030 & 0.030 & 0.499 & 0.030 \\
			DGK+\iforest{} & 0.100 & 0.030 & 0.038 & 0.801 & 0.074 \\
			\entropy & 0.100 & 0.030 & 0.117 & 0.623 & 0.062 \\
			\multi & 0.000 & 0.020 & 0.032 & 0.505 & 0.030 \\
			\bottomrule
		\end{tabular}
		\vspace{0.2in}
	\end{subtable}
	\begin{subtable}{\textwidth}
		\centering
		\caption{\kd dataset}
		\label{tab:anomaly_2_kd}
		\begin{tabular}{lrrrrr}
			\toprule
			\textbf{Method} & Prec@10 & Prec@100 & Prec@1000 & AUC & AP\\
			\midrule
			\method & \textbf{0.800} & \textbf{0.940} & \textbf{0.863} & \textbf{0.716} & \textbf{0.403} \\
			\smt & 0.200 & 0.190 & 0.122 & \underline{0.715} & 0.082\\
			\glocalkd & \underline{0.700} & \underline{0.450} & \underline{0.495} & 0.702 & \underline{0.096} \\
			\gbad & $-$ & $-$ & $-$ & $-$ & $-$ \\
			GF+\iforest & 0.300 & 0.060 & 0.038 & 0.650 & 0.053 \\
			G2V+\iforest{}& $-$ & $-$ & $-$ & $-$ & $-$ \\
			DGK+\iforest{} & 0.100 & 0.090 & 0.068 & 0.644 & 0.061 \\
			\entropy & 0.400 & 0.240 & 0.130 & 0.541 & 0.040 \\
			\multi & 0.000 & 0.040 & 0.035 & 0.498 & 0.030 \\
			\bottomrule
		\end{tabular}
		\vspace{0.2in}
	\end{subtable}
	\begin{subtable}{\textwidth}
		\centering
		\caption{\shsyn dataset}
		\label{tab:anomaly_2_shsyn}
		\begin{tabular}{lrrrrr}
			\toprule
			\textbf{Method} & Prec@10 & Prec@100 & Prec@1000 & AUC & AP\\
			\midrule
			\method & \textbf{0.600} & \textbf{0.640} & \textbf{0.459} & \textbf{0.920} & \textbf{0.350} \\
			\smt & \underline{{0.400}} & \underline{0.450} & 0.162 & 0.720 & 0.163 \\
			\glocalkd & \underline{{0.400}} & 0.440 & 0.321 & 0.874 & 0.289 \\
			\gbad & \underline{{0.400}} & {0.420} & \underline{0.337} & \underline{0.881} & \underline{0.315} \\
			GF+\iforest & 0.200 & 0.250 & 0.140 & 0.831 & 0.127 \\	
			G2V+\iforest{} & 0.000 & 0.020 & 0.021 & 0.465 & 0.020 \\
			DGK+\iforest{} & 0.200 & 0.110 & 0.063 & 0.739 & 0.125 \\
			\entropy & 0.200 & 0.110 & 0.061 & 0.553 & 0.043 \\
			\multi & 0.000 & 0.014 & 0.024 & 0.475 & 0.020 \\
			\bottomrule
		\end{tabular}
	\end{subtable}
\end{table}

\vspace{0.05in}
\noindent
\textbf{Baselines}: We compare \method with:
\cbit
\setlength{\itemindent}{-.1in}

\item 
\smt: A simplified version of \method that uses the Standard Motif Table to encode the graphs. 

\item \textsc{GLocalKD} \cite{ma2022deep}: A recent graph neural network (GNN) based approach that leverages knowledge distillation to train and was shown to achieve good performance and robustness in semi-supervised and unsupervised settings. We used the default setting of 3 layers, 512 hidden dimensions, and 256 output dimensions as recommended in the original paper \cite{ma2022deep}. Independently, we also performed sensitivity analysis on varying the hyper-parameter settings and found that the default one returned top performance.

\item 
\gbad \cite{journals/ida/EberleH07}: The closest existing approach for anomaly detection in node-labeled graph databases (See Sec. \ref{sec:related}). Since it cannot 
handle multi-edges,
we input the $G_j$'s as simple graphs setting all the edge multiplicities to 1. 

\item Graph Embedding + \iforest: We pair different graph representation learning approaches with state-of-the-art outlier detector \iforest\cite{liu2012isolation}, as they  
cannot directly detect anomalies. We consider the following combinations:
\begin{itemize}
	\item Graph2Vec \cite{narayanan2017graph2vec} (G2V)+\iforest: 
	  G2V cannot handle edge multiplicities, thus we set all to 1.
	\item Deep Graph Kernel\cite{yanardag2015deep}(DGK)+\iforest
	\item GF+\iforest: Graph (numerical) features (GF) include number of edges of each label-pair and number of nodes of each label.
\end{itemize}

\item 
\entropy quantifies skewness of the distribution on the non-zero number of edges over all possible label pairs as the anomaly score. 
A smaller entropy implies there exist rare label-pairs and hence higher anomalousness.

\item 
\multi uses sum of edge multiplicities as anomaly score. We tried other simple statistics, e.g., \#nodes, \#edges, their sum/product, which do not perform well.

\ceit

\vspace{0.05in}
\noindent \textbf{Performance measures}: Based on the ranking of graphs by anomaly score, we measure Precision at top-$k$ for $k=\{10,100,1000\}$, and also report Area Under ROC Curve (AUC) and Average Precision (AP) that is the area under the precision-recall curve. Since most of the methods, including \method, are deterministic, we run most of the methods once and report all the measures. For Graph2Vec and Deep Graph Kernel with some randomization, we also run multiple times to check the consistency of performance. Additionally, for those with hyper-parameters, we use the default settings from the corresponding publicly available source codes.

\subsubsection{\bf Detection of path anomalies} We report detection results 
on \sh and \hw datasets in Tables~\ref{tab:anomaly_1_sh} and~\ref{tab:anomaly_1_hw} (performance on \kd is similar and omitted for brevity) and on \shsyn{} in Table~\ref{tab:anomaly_1_shsyn}.

\method consistently outperforms all baselines by a large margin across all performance measures in detecting path anomalies. More specifically, \method provides 16.9\% improvement over the \underline{runner-up} ({underlined}) on average across all measures on \sh, and 10.2\% on \hw.
Note that the runner-up is not the same baseline consistently
across different performance measures. 
Benefits of motif search is evident looking at the superior results over \smt.
G2V+\iforest{} produces decent performance w.r.t. most measures but is still much lower than those of \method. Similar observations are present on \shsyn{}.


\subsubsection{\bf Detection of label anomalies} Tables~\ref{tab:anomaly_2_hw} and~\ref{tab:anomaly_2_kd} report detection results on the two larger datasets, \hw and \kd (performance on \sh is similar and omitted for brevity), and Table~\ref{tab:anomaly_2_shsyn} provides results on \shsyn{}. Note that \gbad{} and G2V+\iforest{} failed to complete within 5 days on \kd{}, thus their results are absent in Table~\ref{tab:anomaly_2_kd}. Note that \kd{} is a relatively large-scale dataset, having both larger and more graphs than the other datasets.

In general, we observe similar performance advantage of \method over the baselines for label anomalies. 
The exceptions are \gbad and \glocalkd{} which perform comparably, and appear to suit better for label anomalies, potentially because changing node labels disrupts structure more than the addition of a few short isolate paths.
\gbad however does not scale to \kd, and the \underline{runner-up} on this dataset performs significantly worse. Similar observations are also seen on \shsyn{} dataset.

\subsection{Case Studies}
\label{ssec:cs}

\textbf{Case 1 - Anomalous transaction records:}
The original accounting databases we are provided with by our industry partner do not contain any ground truth labels.
Nevertheless, they beg the question of whether \method unearths any dubious journal entries that would raise an eyebrow 
from an economic bookkeeping perspective.
In collaboration with their accounting experts, we analyze the top 20 cases as ranked by \method. 
Due to space limit, we elaborate on one case study from each dataset/corporation as follows.




In \sh, we detect a graph with a large encoding length yet relatively few (27) multi-edges, as shown in Fig. \ref{fig:shcases}, consisting of several small disconnected components.
In accounting terms the transaction is extremely complicated, likely the result of a (quite rare) ``business restructuring'' event. In this single journal entry there exist many independent simple entries, involving only one or two operating-expense (OE) accounts, while other edges arise from compound entries (involving more than three accounts).
This event involves reversals (back to prepaid expenses) as well as re-classification of previously booked expenses.
The fact that all these bookings are recorded within a single entry leaves room for manipulation of economic performance and mis-reporting via re-classification, which deserves an audit for careful re-examination.



\hide{ 
We also identify a couple of graphs with node C (cash) at the center, shown in Fig. \ref{fig:shcases}(b).
Here, edges among cash accounts occur typically due to currency transfers (e.g., US\$ to EUR).
Edges from Cash to most SOL nodes are also typical---\sh pays vendors for various expenses for which they receive bills. Edges from Cash to most OE (operating expense) nodes carry economic meaning (SH pays vendors for various expenses), however, these direct links bypass the Accounts Payable stage and can be anomalous. For example, there exist edges from Cash to Cost of Good Sold that are very troubling as most Cost of Good Sold receives inflow from inventory accounts.
}

%
%

\hide{ 
The cluster of the most anomalous graphs depicted in Fig. \ref{fig:lenvsedges} contains 16 graphs that 
are very similar to each other (hence the similar scores), 3 of which are shown in Fig. \ref{fig:shcases}(c).
An edge from an OE (operating expense) node to an ORS (operating revenue/sales) node represents a reduction of both expense and sales amounts. Linking of these two accounts is highly unusual from an economic bookkeeping perspective for two reasons. (This two-node motif only takes place 18 times in \sh.) First, sales and expense cycles are separate within an organization.
Second, the monetary price-levels are different between the two cycles: selling vs. buying prices.
}
\setlength{\intextsep}{2pt}
\setlength{\columnsep}{9pt}%
\setlength{\tabcolsep}{-4pt}
\begin{figure}[!h]
	\centering
	\begin{tabular}{cc}
		\centering
		\includegraphics[width = 0.24\linewidth]{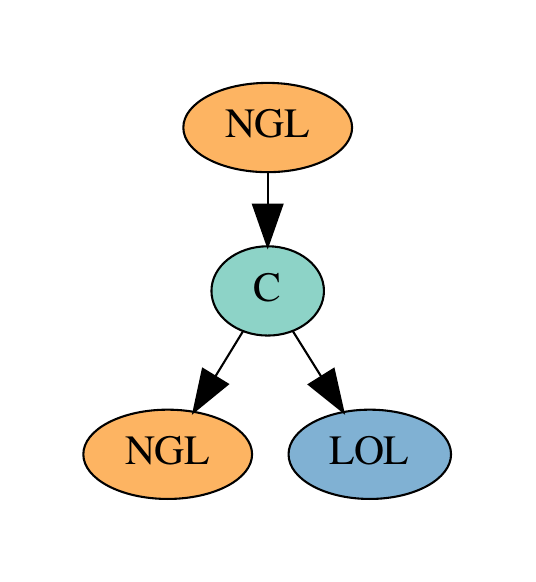}  &
		\includegraphics[width = 0.30\linewidth]{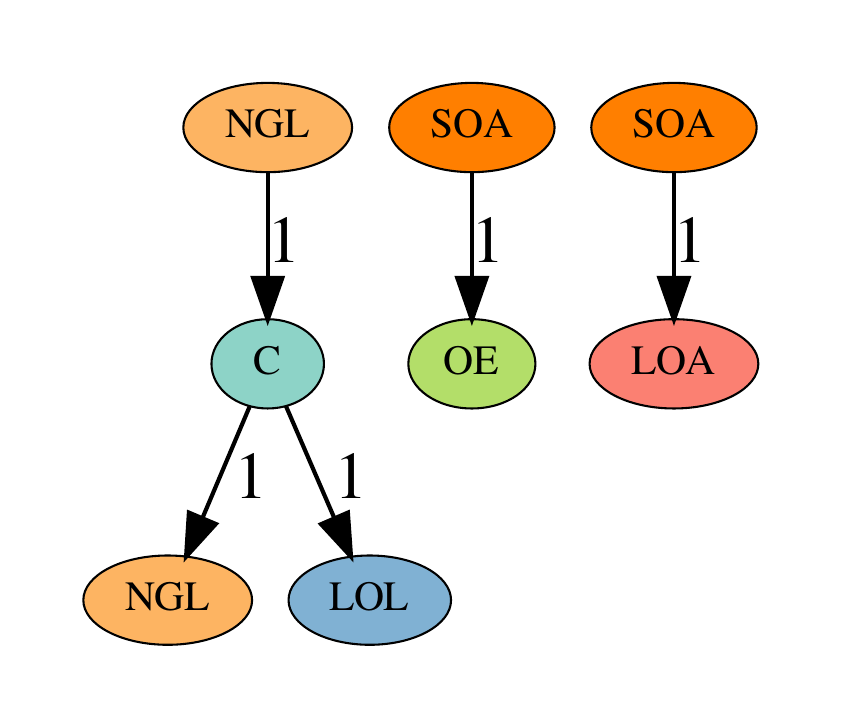}  
	\end{tabular}
	\caption{(left) A rare motif, (right) Anomalous graph in \hw. \label{fig:hwcase}}
\end{figure}
\begin{figure}[!h]
	\centering
	\begin{tabular}{c}
		\begin{tabular}{c}
			\centering
			\includegraphics[width = \linewidth]{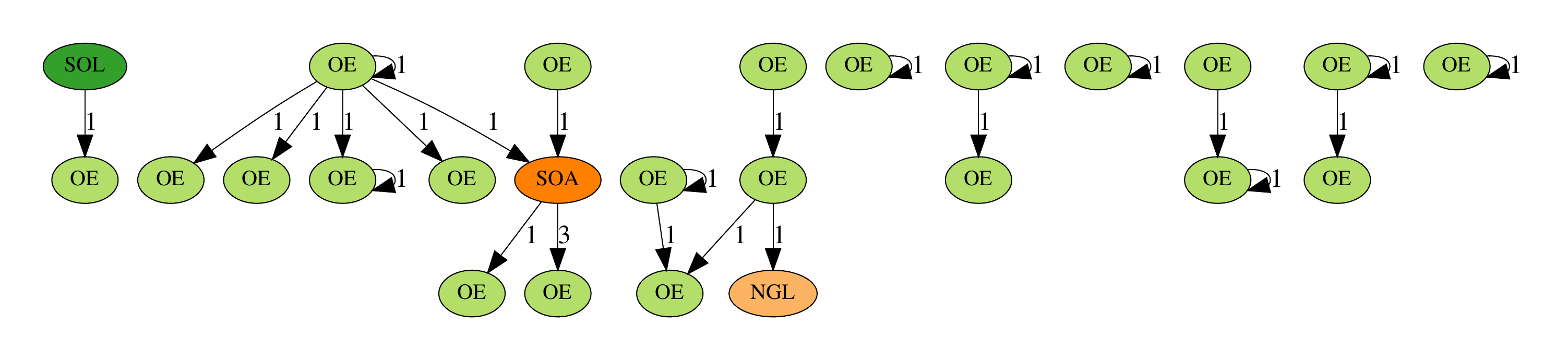} 
		\end{tabular}
	\end{tabular}
	\caption{Anomalous graph in \sh. \label{fig:shcases}}
\end{figure}

In Fig.~\ref{fig:hwcase} (left) we show a motif with sole usage of 1 in the dataset, 
which is used to cover an
anomalous graph (right) in \hw.
Edge NGL (non-operating gains\&losses) to C (cash) depicts an unrealized foreign exchange gain and is quite unusual from an economic bookkeeping perspective.
This is because, by definition, unrealized gains and losses do not involve cash.
Therefore, proper booking of the creation or relinquishment of such gains or losses should not involve cash accounts.
Another peculiarity is the three separate disconnected components, each of which represents very distinct economic transactions: one on a bank charge related to a security deposit,
one on health-care and travel-related foreign-currency business expense (these two are short-term activities),
and a third one on some on-going construction (long-term in nature).
It is questionable why these diverse transactions are grouped into a single journal.
Finally, the on-going construction portion involves reclassifying a long-term asset into a suspense account, which requires follow-up attention and final resolution.

\begin{wrapfigure}[9]{r}{0.33\linewidth}
	\vspace{-0.15in}
	\centering
	\begin{tabular}{cc}
		\centering
		\includegraphics[width = 0.41\linewidth]{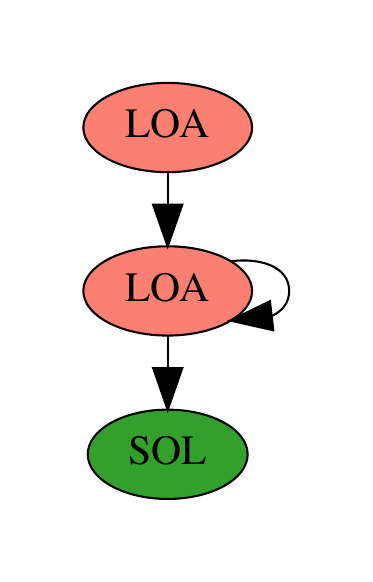} &
		\includegraphics[width = 0.42\linewidth]{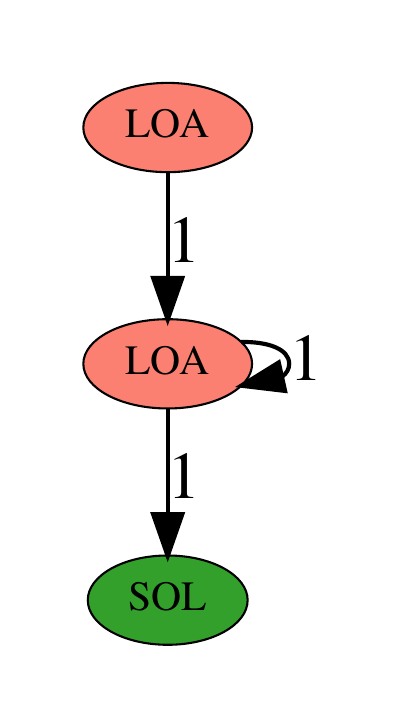}  
	\end{tabular}
	\vspace{-0.15in}
	\caption{(left) A rare motif, (right) Anomalous graph in \kd. \label{fig:kdcase}}
\end{wrapfigure} 

Finally, the anomalous journal entry from \kd~involves the motif shown in Fig. \ref{fig:kdcase} (left) where the corresponding graph is the exact motif with multiplicity 1 shown on the (right). This motif has sole usage of 1 in the dataset and is odd from an accounting perspective. 
Economically, it represents giving up an existing machine, which is a long-term operating asset (LOA), 
in order to reduce a payable or an outstanding short-term operating liability (SOL) owed to a vendor.
Typically one would sell the machine and get cash to payoff the vendor with some gains or losses. We also note that the $\mt$ does not contain the 2-node motif LOA$\rightarrow$SOL.
The fact that it only shows up once, within single-usage motif, makes it suspicious.

Besides the quantitative evidence on detection performance, these number of case studies provide qualitative support to the effectiveness of \method in identifying anomalies of interest in accounting domain terms, worthy of auditing and re-examination.

\textbf{Case 2 - Enron scandal:} We study the correlation between \method's anomaly scores of the daily email exchange graphs 
and the major events in Enron's timeline. 
Fig. \ref{fig:enron} shows that days with large anomaly scores mark drastic discontinuities in time, which
coincide with important events related to the financial scandal\footnote{
	\url{http://www.agsm.edu.au/bobm/teaching/BE/Enron/timeline.html}, \url{https://www.theguardian.com/business/2006/jan/30/corporatefraud.enron}}.
It is also noteworthy that the anomaly scores follow an increasing trend over days, capturing the escalation of events up to key personnel testifying in front of Congressional committees. 

\setlength{\intextsep}{0pt}
\setlength{\tabcolsep}{4pt}
\begin{figure*}[!t]
	\begin{minipage}{\linewidth}
		\centering \includegraphics[width = \linewidth]{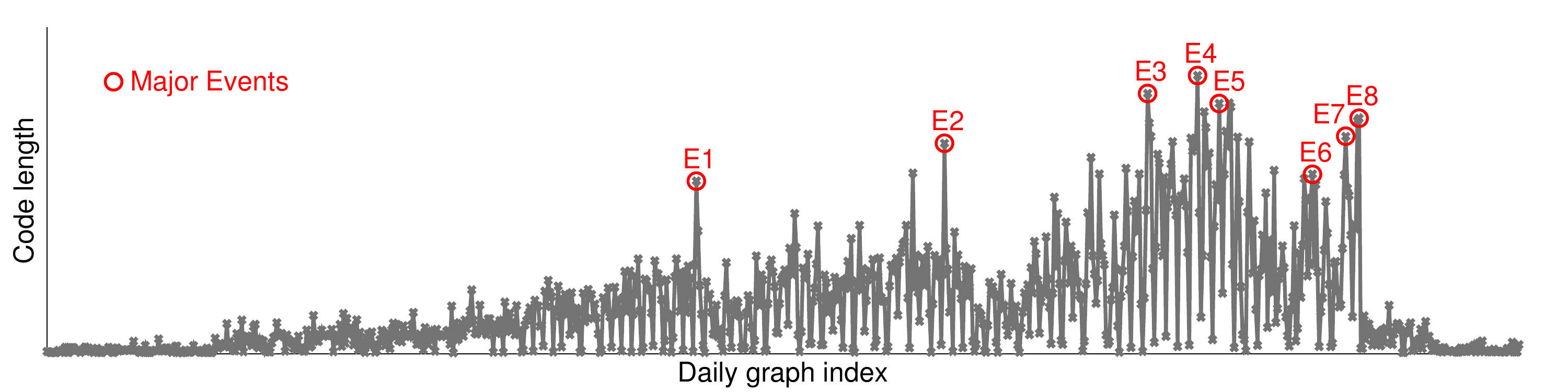} 
	\end{minipage}
	\centering
	\begin{tabular}{l|p{0.15\textwidth}|p{0.75\textwidth}}
		\toprule
		E1 & Dec 13, 2000 & Enron announces that president and COO Jeffrey Skilling will take over as chief executive in February. 
		\\
		E2 & May 23, 2001 & Enron completes its 1,000,000-th transaction via Enron Online. \\
		E3 & Sep 26, 2001 & Employee Meeting. Kenneth Lay tells employees: Enron stock is an ``incredible bargain'' and  ``Third quarter is looking great.'' 
		\\
		E4 & Oct 24-29, 2001 & Enron sacks Andrew Fastow.  In vain 
		Lay calls 
		chairman of the Fed and the Treasury and Commerce secretaries 
		to solicit help.
		\\
		E5 & Nov 9, 2001 & Dynegy agrees to buy Enron for about \$9 billion in stock and cash. \\
		E6 & Jan 10, 2002 & DOJ confirms criminal investigation begun. Arthur Andersen announces employees in Houston Division had destroyed documents. \\
		E7 & Jan 30, 2002 & Enron names Stephen F. Cooper new CEO.\\
		E8 & Feb 07, 2002 & A. Fastow 
		invokes the Fifth Amendment before Congress; J. Skilling testifies he knew of no wrongdoing at Enron when he resigned.
		\\
		\bottomrule
	\end{tabular}
\vspace{-0.1in}
	\caption{
		Code lengths of \enron's daily email exchange graphs. Large values coincide with key events of the financial scandal.}
	\label{fig:enron}
\end{figure*}
\begin{figure}[H]
	\centering
	\begin{tabular}{p{1cm}cc}
		\centering
		& \includegraphics[width = 0.35\linewidth]{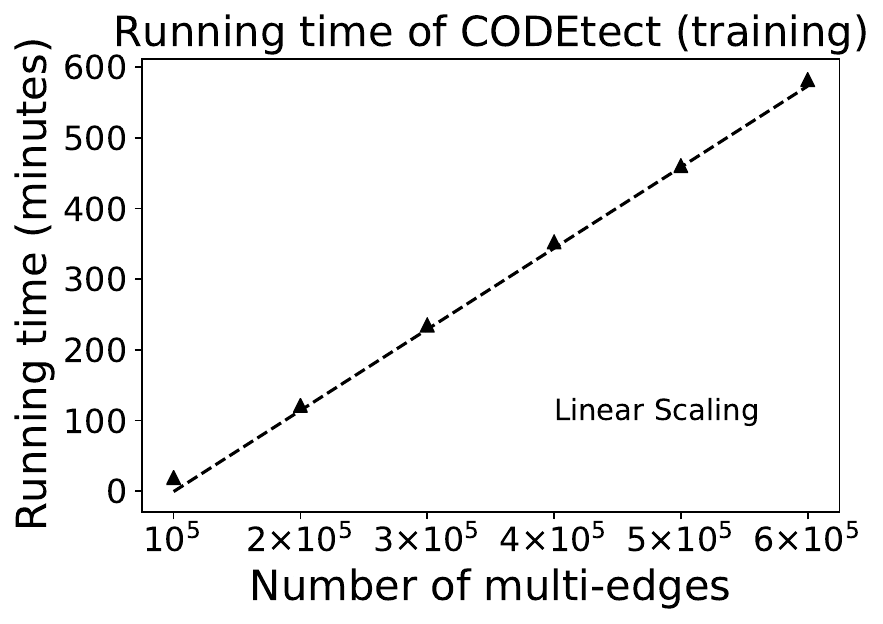}  &
		\includegraphics[width = 0.37\linewidth]{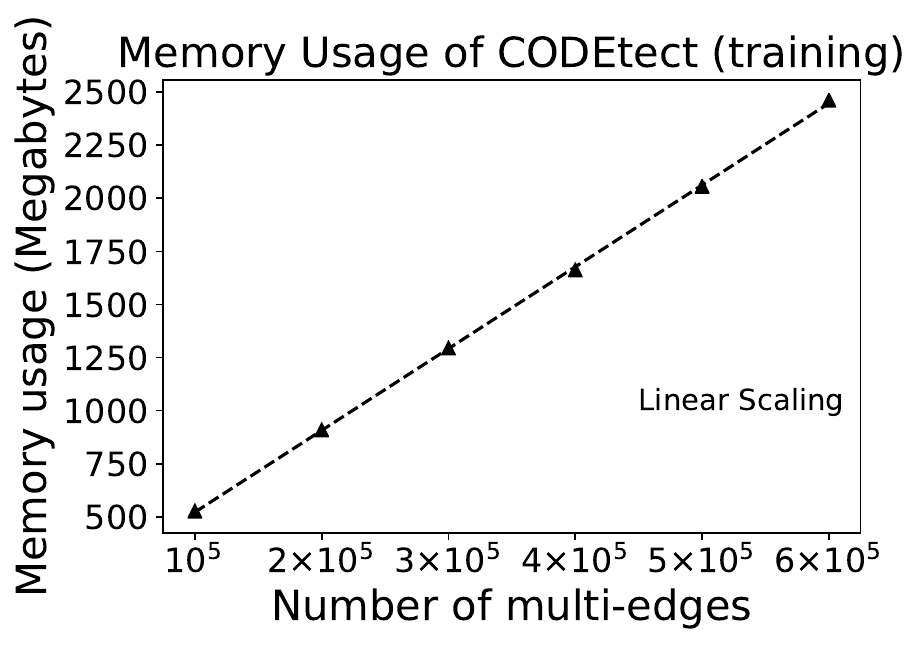}  \\
		& \includegraphics[width = 0.35\linewidth]{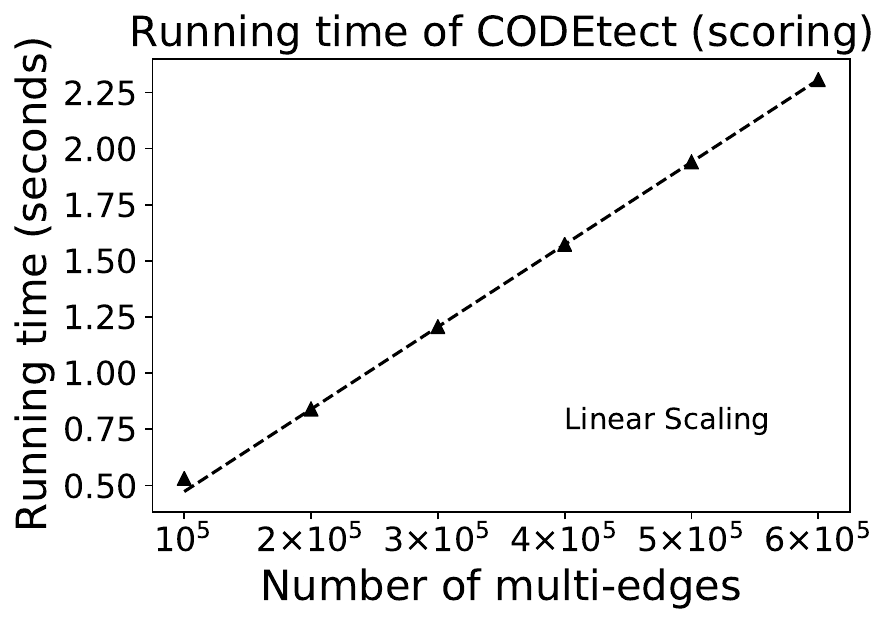}  &
		\includegraphics[width = 0.37\linewidth]{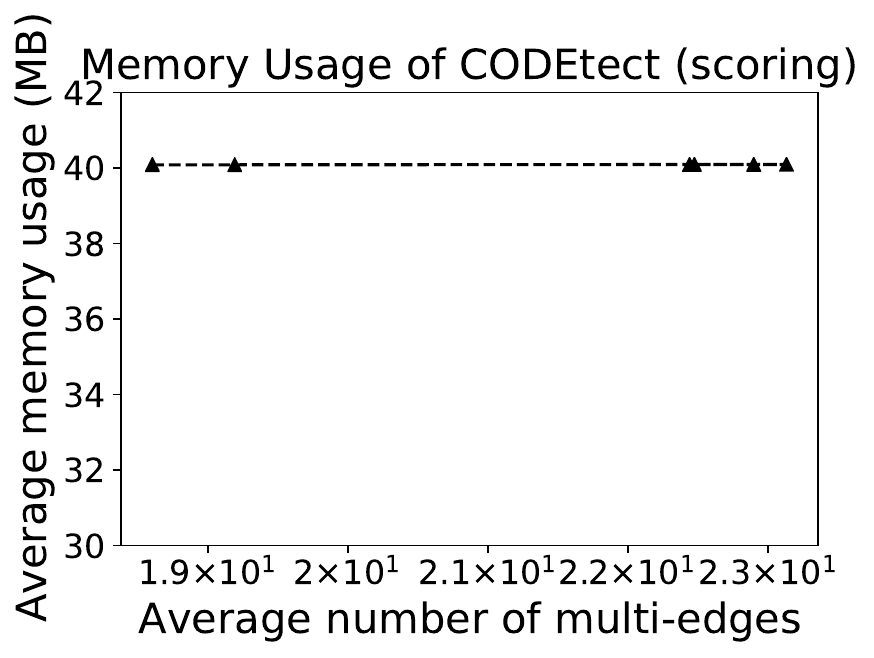}  
	\end{tabular}
	\vspace{-0.1in}
	\caption{(left) Time, (right) Memory consumption (training refers to running the full motif table search, while scoring assumes the motif table has been learned and only executes graph encoding for anomaly scoring.) Note that \method{} training and scoring both scale linearly in the number of multi-edges w.r.t. both time and memory usage.}
	\label{fig:scale}
	\vspace{-0.15in}
\end{figure}

\subsection{Scalability}
To showcase the scalability of \method{}, in regard to running time and memory consumption, we randomly selected subsets of graphs in \kd database with different sizes, i.e., $\{40, 45, 50, 55, 60, 65, 70, 75, 80, 85, 90, 95, 100\} \times 10^3$. We re-sample each subset of graphs three times and report the averaged result for each setting. A summary of the results is presented in Figure~\ref{fig:scale}. We observe a linear scaling of \method with increasing size of input graphs as measured in number of multi-edges with respect to both time and memory usage.

\section{Conclusion}
\label{sec:conclusion}

We introduced \method, (to our knowledge) the first graph-level anomaly detection method for \textit{node-labeled multi-graph databases}; which appear in numerous real-world settings such as social networks and financial transactions, to name a few. 
The main idea is to identify key network motifs that encode the database concisely and employ compression length as the anomaly score.
To this end, we presented (1) novel lossless encoding schemes and (2) efficient search algorithms.
Experiments on transaction databases from three different corporations quantitatively showed  that \method significantly outperforms the prior and more recent GNN based baselines across datasets and performance metrics. Case studies, including the Enron database, presented qualitative evidence to \method's effectiveness in spotting instances that are noteworthy of auditing and re-examination.


\section*{Acknowledgments}{
	
	This work is sponsored by NSF CAREER 1452425 and the PwC Risk and Regulatory Services Innovation Center at Carnegie Mellon University. Any conclusions expressed in this material are those of the author and do not necessarily reflect the views, expressed or implied, of the funding parties.
	
}

\bibliographystyle{ACM-Reference-Format}
\bibliography{BIB/refs}

%
%
%
%
%
%
%
%

\end{document}